\documentclass[12pt,reqno]{amsart}

\usepackage[numbers]{natbib}

\usepackage{fix-cm}
\usepackage[utf8]{inputenc}
\usepackage[american]{babel}
\usepackage[babel]{microtype}
\usepackage{amsthm}
\usepackage{amsmath}
\usepackage{amsfonts}
\usepackage{amssymb}
\usepackage{bm}
\usepackage{mathtools}
\usepackage{dsfont}
\usepackage{graphicx}
\usepackage{enumitem}
\usepackage{float}
\usepackage{xcolor}
\usepackage{physics}
\usepackage{subcaption}

\usepackage[colorlinks=true, allcolors=blue, urlcolor=blue]{hyperref}
\usepackage[left=2cm, right=2cm, top=2.5cm]{geometry}

\numberwithin{equation}{section}

\makeatletter
\patchcmd{\@maketitle}
  {\ifx\@empty\@dedicatory}
  {\ifx\@empty\@dedicatory \par\vskip0.25em{\begin{center}\scriptsize\@setaddresses\end{center}}}
  {}{}
\makeatother

\newtheorem{proposition}{Proposition}

\newtheorem{theorem}{Theorem}

\newtheorem{corollary}{Corollary}

\theoremstyle{definition}

\newtheorem{definition}{Definition}
\newtheorem{example}{Example}

\renewcommand{\real}{\mathbb R} %real
\newcommand{\hi}{\mathcal{H}} %Hilbert space H
\newcommand{\hik}{\mathcal{K}} %Hilbert space K
\newcommand{\lh}{\mathcal{L(H)}} %bounded linear operators
\newcommand{\lsh}{\mathcal{L}_s(\mathcal{H})} %selfadjoint linear operators
\newcommand{\lk}{\mathcal{L(K)}} %bounded linear operators on K
\newcommand{\sh}{\mathcal{S(H)}} %states on H
\newcommand{\eh}{\mathcal{E(H)}} %effects
\newcommand{\kb}[2]{|#1\rangle\langle#2|} %ketbra
\renewcommand{\tr}[1]{\mathrm{tr}\left[#1\right]} %trace
\newcommand{\ptr}[2]{\mathrm{tr}_{#1}\left[#2\right]} %partial trace
\newcommand{\id}{\mathds{1}} %identity operator
\renewcommand{\varrho}{\rho}

\newcommand{\meters}{\mathcal{M}}
\newcommand{\A}{\mathsf{A}}%generic observable
\newcommand{\B}{\mathsf{B}}%generic observable
\newcommand{\Eo}{\mathsf{E}}%generic observable
\renewcommand{\L}{\mathcal{L}}%set of linear maps
\newcommand{\X}{\mathsf{X}}%generic observable
\newcommand{\Z}{\mathsf{Z}}

\newcommand{\E}{\mathcal{E}}

\newcommand{\R}{\mathcal{R}}
\newcommand{\I}{\mathcal{I}}

\newcommand{\J}{\mathcal{J}}
\newcommand{\G}{\mathcal{G}}
\newcommand{\ins}{\mathrm{Ins}}

\newcommand{\ch}{\mathrm{Ch}}

\newcommand{\luders}[1]{\mathcal{L}^{#1}}

\title{Adversarial Information Gain in Non-ideal Quantum Measurements}
\author{Andrés Muñoz-Moller \and Leevi Leppäjärvi \and Teiko Heinosaari}
\address{Faculty of Information Technology, University of Jyväskylä, Finland}

\begin{document}

\begin{abstract}

Performing a quantum measurement yields two different results: a classical outcome drawn from a probability distribution, according to Born's rule, and a quantum outcome corresponding to the post-measurement state. Quantum devices that provide both outcomes can be described through quantum instruments. In a realistic scenario, one can expect that the observer's obtained classical and quantum outcomes are non-ideal: this can be due to experimental limitations, but could also be explained by adversarial interference, that is, a second party that disturbs the device through a concealed measurement to obtain information. The second scenario can be interpreted through quantum compatibility, as it implies that both the observer's instrument and the adversary's measurement can be performed simultaneously. In this work, we show how the noise of the observer's device relates to the amount of information that the adversary can obtain. We study scenarios in which the adversary aims to acquire information on the same basis as the observer's measurement, or on a mutually unbiased basis with respect to the observer's basis. In both cases, we derive necessary and sufficient conditions for the compatibility of a single qubit non-ideal quantum instrument and a noisy meter, from which we obtain the maximum amount of information that the adversary can extract in terms of the noise parameters of the observer's instrument. Finally, we provide the device implementation from the adversary's point of view for the same basis scenario.

\end{abstract}

\maketitle
\makeatletter \let\@setaddresses\relax \makeatother

%%%%%%%%%%%%%%%%%%%%%%%%%%%%%%%%%%%%%%%%%%%%%%%%%%%%
\section{Introduction}
%%%%%%%%%%%%%%%%%%%%%%%%%%%%%%%%%%%%%%%%%%%%%%%%%%%%

A quantum measurement serves two functions.
First, it provides information about the input state through the probability distribution of measurement outcomes.
Second, the act of measurement inevitably disturbs the input state, transforming it into a new state.
This transformation may be undesirable, for instance, when the goal is to learn about the system while minimally disturbing it, or it can be intentionally exploited as a means of state preparation.
The concept of a quantum instrument unifies these two aspects into a single mathematical object that describes both the information extraction and state transformation.
In particular, when we have access to a quantum device and can observe its input–output behavior but not its internal workings, the appropriate mathematical description is given by a quantum instrument \cite{Davies1970}.

Suppose that an observer operates a device that is described as an instrument $\mathcal{I}$ and is hence familiar with its input–output behavior.
In a realistic setting, the instrument is not ideal but contains some noise.
In the present investigation, we consider the scenario where this noise originates from the deliberate actions of an adversarial agent.
This adversary may act covertly, and its identity or even its existence might be unknown to the observer. The adversary’s goal is to exploit the measurement process to gain additional information about the quantum state being measured, potentially without detection.
In the worst case scenario, the adversary has designed the instrument in a particular way and keeps a part of it to herself.
The central question is: 
\emph{to what extent can such an adversary gain information about the state of the system being measured?}
We will present a framework where this question can be studied, and fully answer it in various representative cases.

To get an idea of the role of noise in this question, let us first recall a common class of measurements called basis measurements, also known as sharp or projective measurements.
In this idealized case, each measurement outcome $i$ corresponds to a one-dimensional projection operator $P_i$.
The probability of obtaining outcome $i$ for an input state $\varrho$ is given by
$p(i) = \operatorname{tr}[\varrho P_i]$,
and the corresponding conditional output state is the projection $P_i$.
Because the measurement leaves the system in an eigenstate of the measured observable, the same measurement can be repeated immediately with the same outcome.
Hence, it is easy to see that an adversary could have implemented an identical projective measurement without leaving any detectable traces.
Since the statistics and post-measurement states would be exactly the same, no observation on the output system could reveal whether the measurement was performed by the intended observer or by an eavesdropper.
Also, an adversary could have performed a coarse-grained version of the same measurement and this, too, would remain undetectable.
These are, however, the only possibilities if the observer implements the basis measurement.
We therefore conclude that, in the case of basis measurements, an adversary can acquire at most the same information as the observer, but never more.

In realistic scenarios, measurements are noisy to some extent. To account for this, the meter must be described more generally by a positive operator-valued measure (POVM).
A natural and commonly used instrument associated with a meter $\A$ is the Lüders instrument \cite{Luders1950}, defined by the conditional state transformations
$\varrho \mapsto \sqrt{\A_i}\, \varrho \, \sqrt{\A_i}$.
When we look only at the measurement statistics, i.e., meter $\A$ itself, we cannot exclude the possibility that some external agent has interacted with the system.
However, the Lüders instrument corresponding to an unsharp meter is still regarded as ideal, at the level of state transformations.
This is because the Lüders instrument represents the minimally disturbing implementation of the given meter: among all possible instruments that yield the same measurement statistics, the Lüders instrument causes the smallest possible state change compatible with those statistics \cite{HayashiBook06}.
In other words, it models an implementation that extracts the prescribed amount of information but introduces no additional, extraneous disturbance.
Hence, although the measurement outcome distributions may be noisy, the Lüders instrument is optimal relative to that unsharp meter. This is also reflected by the fact that the adversary can obtain only information that can be classically post-processed from the meter $\A$ \cite{Leppjrvi2024}. Thus also in this case an adversary can acquire at most the same information as the observer.

In the present work, we investigate a class of non-ideal quantum instruments that deviate from the ideal device both in their measurement outcome statistics and in their state transformation properties, as shown in Fig.~\ref{fig:noisy_instrument}.
In this model, the meter corresponds to a noisy version of a basis measurement, while the associated state transformations are described by noisy Lüders-type operations.
Such instruments represent a realistic setting where the measurement device is imperfect, either due to inherent physical noise or to deliberate interference.
In the latter case, the adversary can obtain non-trivial information about the input state while keeping her presence undetectable.  
Her actions appear merely as part of the instrument’s non-ideal functioning.
It is therefore vital that the observer can estimate the potential information gain of the adversary from the input-output data of the instrument.
This is exactly the goal of the present investigation.

\begin{figure}[h]
    \centering
    \includegraphics[width=0.75\linewidth]{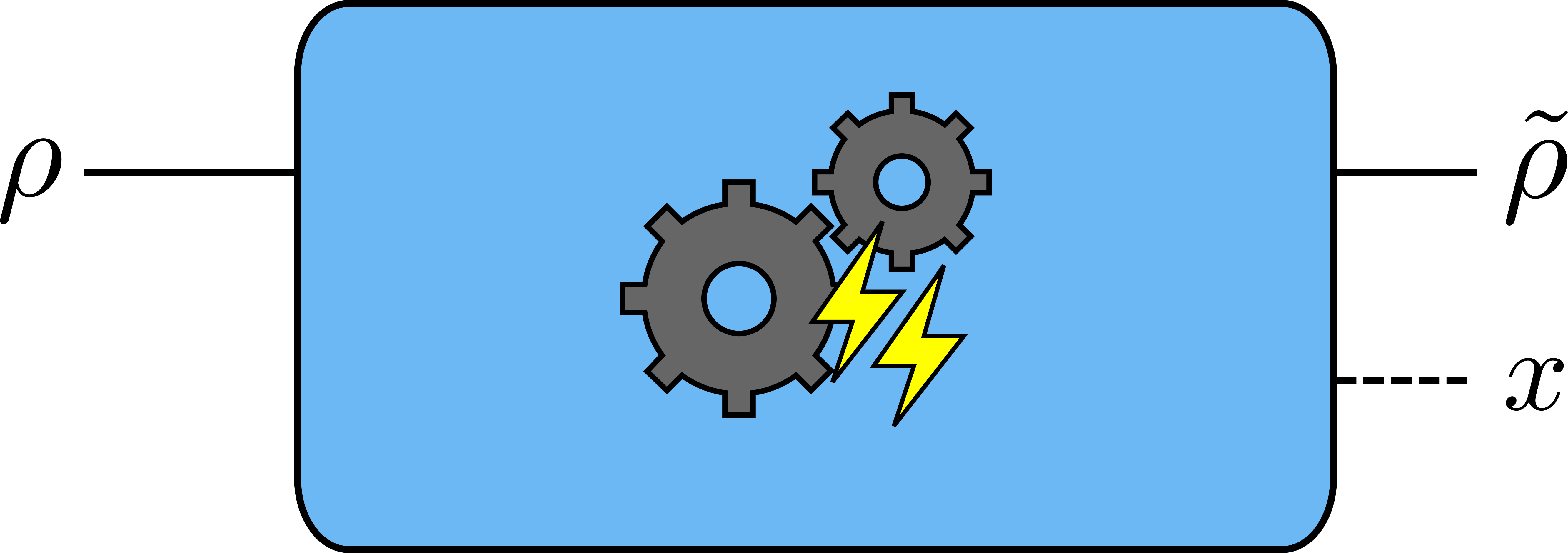}
    \caption{Device that performs a non-ideal quantum measurement on state $\rho$. The outputs correspond to the classical measurement outcome $x$ and the quantum post-measurement state $\tilde{\rho}$, which deviate from an ideal measurement due to noise.}
    \label{fig:noisy_instrument}
\end{figure}

Two special cases are of particular interest:
\begin{itemize}
\item[1.] \emph{Aligned information retrieval}: the adversary seeks information in the same basis as that of the observer’s measurement.
\item[2.] \emph{Complementary information retrieval}: the adversary aims to extract information in a complementary basis, i.e., a basis that is mutually unbiased with respect to the observer’s measurement basis.
\end{itemize}

These two cases have certain similarities, but also crucial differences.
Our two main findings in the above cases are:
\begin{itemize}
\item[1.] When the adversary retrieves information in the same basis as the observer, he can obtain more information than the observer.
The higher the precision of the observer’s meter, the greater the maximal precision attainable by the adversary.
\item[2.] When the adversary retrieves information in a basis that is complementary to that of the observer, the situation is reversed:
The higher the precision of the observer’s meter, the lower the maximal precision attainable by the adversary.
\end{itemize}

In both of these cases we derive the exact formula for the maximal precision of the adversarial information gain.
The technical results behind these findings are related to incompatibility regions of certain type of meters and instruments.
Our contribution is hence also to derive novel incompatibility inequalities for certain classes of meters and instruments.

%%%%%%%%%%%%%%%%%%%%%%%%%%%%%%%%%%%%%%%%%%%%%%%%%%%%
\section{Modelling adversarial information acquisition in nonideal measurements}
%%%%%%%%%%%%%%%%%%%%%%%%%%%%%%%%%%%%%%%%%%%%%%%%%%%%

A quantum measurement has several mathematical descriptions, depending on what functional features of it we are interested in. An instrument gives a description of both measurement outcome statistics and conditional state transformation, without going into details of the inner details of the measurement. 
Hence, one can think of it as a full black-box description of a measurement.
In this sense, an instrument provides a complete black-box description of a measurement process at the operational level.
This viewpoint matches our setting: we assume that the observer has full knowledge of how the measurement device acts on input states and what outputs it produces, but possesses no information about the underlying mechanisms by which these transformations are realized. Let us formalize the different descriptions (see e.g. \cite{Heinosaari2011,Busch2016}).

Let $\hi$ be a finite $d$-dimensional Hilbert space.
We denote the set of linear operators on $\hi$ by $\lh$ and the subset of selfadjoint (Hermitian) operators by $\lsh$. 
Quantum states are then represented by density operators on $\hi$, i.e., positive operators with trace equal to one. The set of density operators on $\hi$ is denoted by $\sh$. 
A quantum operation (in the Schrödinger picture) is a linear map $\Phi: \lh \to \lk$ that is completely positive (CP), i.e., $\Phi \otimes id_\mathcal{V}$ is a positive map for all (finite-dimensional) ancillary Hilbert spaces $\mathcal{V}$, where $id_\mathcal{V}$ denotes the identity map on $\mathcal{L(V)}$, and trace nonincreasing (TNI), i.e., $\tr{\Phi(X)} \leq \tr{X}$ for all $X \in \lh$. 
A quantum operation $\Phi$ is called a channel if it is trace-preserving (TP) so that $\tr{\Phi(X)} = \tr{X}$ for all $X \in \lh$. 
In the Heisenberg picture, this is equivalent to the adjoint map $\Phi^*$ being unital, i.e., $\Phi^*(\id_\hik) = \id_\hi$. A quantum operation $\Phi$ can be represented in terms of Kraus operators $\qty{K_i}_i$ such that $\Phi(X) = \sum_i K_i X K_i^*$ and $\sum_i K_i^* K_i \leq \id$, where the equality holds if $\Phi$ is a channel \cite{Kraus1983-qe}. The set of quantum channels from $\lh$ to $\lk$ is denoted by $\ch(\hi,\hik)$. 
Quantum channels describe the physical transformations of states, whereas quantum operations can be interpreted as probabilistic transformations, where the probability that a transformation given by an operation $\Phi$ that is applied to a state $\varrho$ is successful is given by $\tr{\Phi(\varrho)}$.

A quantum instrument from $\lh$ to $\lk$ with $n$ outcomes is a map $\I: x \mapsto \I_x$ from the set of outcomes $[n]\coloneq\{1,\ldots,n\}$ to the set of quantum operations from $\lh$ to $\lk$ such that $\sum_{x=1}^n \I_x$ is a quantum channel in $\ch(\hi,\hik)$.
The physical interpretation is that the probability of detecting an outcome $x \in [n]$, when a system that was in a state $\varrho \in  \sh$, is $\tr{\I_{x}(\varrho)}$. 
After recording $x$, the (unnormalized) post-measurement state of the system is $\I_x(\varrho)$.
 The set of $n$-outcome instruments from $\lh$ to $\lk$ is denoted by $\ins_n(\hi,\hik)$ and the set of all instruments from $\lh$ to $\lk$ by $\ins(\hi,\hik)$. 

In some situations, we may be only interested in measurement outcome probabilities, or have no access to post-measurement states. 
In this case, the relevant concept is that of a meter.
A meter is described by a positive operator-valued measure (POVM) $\A$ from the set of outcomes $[n]$ to the set of effects $\eh$ (i.e. positive operators below the identity) such that $\sum_{x=1}^n \A(x) = \id$. 
The probability of detecting an outcome $x \in [n]$, when a system that was in a state $\varrho \in  \sh$ is measured with a meter $\A$ is then given by the Born rule as $\tr{\A(x)\varrho}$. The set of $n$-outcome meters on $\hi$ is denoted by $\meters_n(\hi)$, and the set of all meters on $\hi$ by $\meters(\hi)$.
Since an instrument contains the descriptions for both the measurement statistics as well as the state transformation, an instrument $\I$ induces a meter $\A^\I$ defined as $\tr{\A^\I(x) \varrho} = \tr{\I_x(\varrho)}$ for all $\varrho \in \sh$, and a channel $\Phi^\I \in \ch(\hi,\hik)$ as $\Phi^\I = \sum_{x=1}^n \I_x$.

\begin{example}[Noisy meters]\label{ex:noise}
     Let $\A$ be a meter with $n$ outcomes. 
     For each $t \in [0,1]$ we can define the depolarizing channel $\mathcal{D}_t$ as $\mathcal{D}_t(X)  =t X +(1-t) \tr{X}\tfrac{1}{d}\id$, where $d = \dim(\hi)$. In the Heisenberg picture we see how $\mathcal{D}_t^*$ transforms the meter $\A$:
     \begin{align}
         \A_t(x)\coloneq\mathcal{D}_t^*(\A(x)) = t \A(x) +(1-t) \tr{\A(x)} \frac{\id}{d}
     \end{align}
     for all $x \in [n]$. We can now interpret $\A_t$ as a noisy meter that mixes the meter $\A$ with depolarizing noise. Here $1-t$ is the noise parameter which indicates how much depolarizing noise is mixed to the original meter. For $t=1$ we have $\A_1 = \A$ and for $t=0$ we have a completely noisy meter. We note that $\mathcal{D}^*_{t'} \circ \mathcal{D}^*_t = \mathcal{D}^*_{tt'}$ for all $t,t' \in [0,1]$ so that the amount of noise can always be added by applying another depolarizing channel. 
     
     In the case when we have an unbiased meter, i.e., when $\tr{\A(x)} = \tfrac{d}{n}$ for all $x \in [n]$, we can model the noise classically via postprocessing the outcomes of the meter. We say that a meter $\A \in \meters_n(\hi)$ can be post-processed into a meter $\B\in \meters_m(\hi)$ if there exists a stochastic matrix $\nu$ consisting of conditional probabilities $\nu_{y|x} \in [0,1]$ for which $\sum_{y \in [m]} \nu_{y|x} = 1$ for all $x \in [n]$ such that $\B(y) =\sum_{x \in [n]} \nu_{y|x} \A(x)$ for all $y\in [n]$. In this case we denote $\B = \nu \circ \A$. Now for $m=n$ if we have $\tr{\A(x)} = \frac{d}{n}$, we can set $\nu^t_{x|y} = t \delta_{x,y} + (1-t)/n$ for all $x,y \in [n]$, where $\delta_{x,y}$ denotes the Kronecker delta. Now clearly
     \begin{align}
        (\nu^t \circ \A)(x) = \sum_{y=1}^n \nu^t_{x|y} \A(y) = t \A(x) +(1-t) \frac{\id}{n} = \A_t(x)
     \end{align}
     for all $x \in [n]$. Furthermore, also in this case now $\nu^{t'} \circ \nu^t = \nu^{t t'}$ so that for unbiased meters the amount of noise can always be added classically.
\end{example}

\begin{example}[Measure-and-prepare instruments]
    Let $\A$ be a meter with $n$ outcomes. One simple way to induce a transformation that still realizes the measurement of $\A$ is to build a conditional state preparator which prepares new quantum states based on the observed measurement outcomes. Thus, for any set of $n$ states $\xi = \{\xi_x\}_{x=1}^n \subset \sh$ we can define a measure-and-prepare instrument related to $\A$ as 
    \begin{align}
        \E^{\A,\xi}_x(\varrho) = \tr{\A(x) \varrho} \xi_x
    \end{align}
    for all $x \in [n]$ and $\varrho \in \sh$. Such a measurement process can retain only classical information about the input states in the post-measurement description: at best the measurement outcome can be encoded into some set of states that can be perfectly discriminated so that the measurement outcome can be obtained from the post-measurement state. On the other hand, we can also lose all this classical information in the transformation if we for example choose $\xi_x = \xi_y$ for all $x \neq y$ with the most unbiased choice being $\xi_x = \tfrac{1}{d}\id$ for all $x \in [n]$, which can be seen as a measurement process of $\A$ that causes the most disturbance to the input state. The ability of only being able to transmit classical information is also reflected by the fact that a channel $\Phi \in \ch(\hi)$ is entanglement-breaking, i.e., $(\Phi \otimes id_\hi)(\varrho)$ is a separable state for all $\varrho \in \mathcal{S}(\hi \otimes \hi)$, if and only if it is of the measure-and-prepare form $\Phi(\varrho) = \sum_{x \in [n]} \tr{\A(x)\varrho} \xi_x$ for all $\varrho \in \sh$ for some states  $\xi = \{\xi_x\}_{x=1}^n \subset \sh$ and some meter $\A\in \meters_n(\hi)$ \cite{Horodecki2003}. 
\end{example}

\begin{example}[Lüders instrument]
    Another way to measure $\A$ is to implement so-called Lüders instrument $\luders{\A}$ defined as 
    \begin{align}
        \luders{\A}_x(\varrho) = \sqrt{\A(x)} \varrho \sqrt{\A(x)}
    \end{align}
    for all $x \in [n]$ and $\varrho \in \sh$. While measure-and-prepare instruments can be seen as the most disturbing way to measure $\A$ from the state transformation perspective, the Lüders instrument can be seen as the least disturbing way to measure $\A$: any instrument $\I \in \ins_n(\hi)$ for which $\A^\I = \A$ is of the form $\I_x = \Phi^{(x)} \circ \luders{\A}_x$ for all $x \in [n]$ for some set of channels $\{\Phi^{(x)}\}_{x \in [n]} \subset \ch(\hi)$. Thus, any measurement process that measures $\A$ can always be obtained by first realizing the Lüders instrument of $\A$ followed by some conditional channels which induce further state change and possible noise.
\end{example}

A quantum instrument can have noise both in information retrieval and state transformation parts of it.
In the current investigation, we assume that this noise is caused by an adversary, who has some kind of way to access the implementation of the instrument that the observer is using.
Such a situation could arise, for instance, when the observer is using the instrument remotely by sending it quantum states and receiving the measurement data and the post-measurement states. Then, from the security point of view, the observer must assume that any noise that they detect could be a result of some adversary action on the instrument or the communication channels between the device and the observer.
In our approach, we do not have to make assumptions on how the adversary is retrieving information.
The key observation is that whatever information the adversary is getting, her device retrieving the information must be compatible with the observer's device. 
In fact, they are after all parts of a single device, and that is exactly the definition of compatibility (see e.g. \cite{Heinosaari2016,Guhne2023}). In this work we assume that the adversary is trying to get classical information about the input states, i.e., they are looking to implement some meter which measures some property of the system in such a way that their actions cannot be detected from the outputs of the original instrument.
Therefore, in order to inspect the information acquisition of the adversary, we need to characterize meters that are compatible with the observer's instrument \cite{Leppjrvi2024}:

\begin{definition}
    A meter $\A \in \meters_n(\hi)$ is compatible with an instrument $\I \in \ins_m(\hi)$ if there exists a joint instrument $\G \in \ins_{nm}(\hi)$ such that 
    \begin{align}
        \I_x = \sum_{y} \G_{(x,y)}, \quad \A(y) = \sum_{x} \A^{\G}(x,y)
    \end{align}
    for all $x \in [m]$ and $y \in [n]$.
\end{definition}

Thus, the compatibility of a meter and an instrument means that both of them can be implemented together with a single instrument which gives a pair of outcomes, one related to the meter and the other related to the instrument such that if we marginalize the post-measurement state over the first outcome we retrieve the post-measurement state of the original instrument and if we marginalize over the second outcome we can retrieve the measurement statistics of the meter, as shown in Fig.~\ref{fig:obs_adv_1}. We note that in the adversarial scenario, it is the ignorance about the adversary's measurement outcome (and even about the existence of the adversary's measurement) that corresponds to such marginalization. When the observer does not have access to the second measurement outcome, the device acts not as the joint instrument but as the marginalization of the joint instrument over the adversary's measurement outcome. Therefore, in no circumstances can the observer meaningfully distinguish between the measuring device being inherently noisy or having an extra classical outcome feeding information to an adversary.

\begin{figure}[h]
    \centering
    \includegraphics[width=0.75\linewidth]{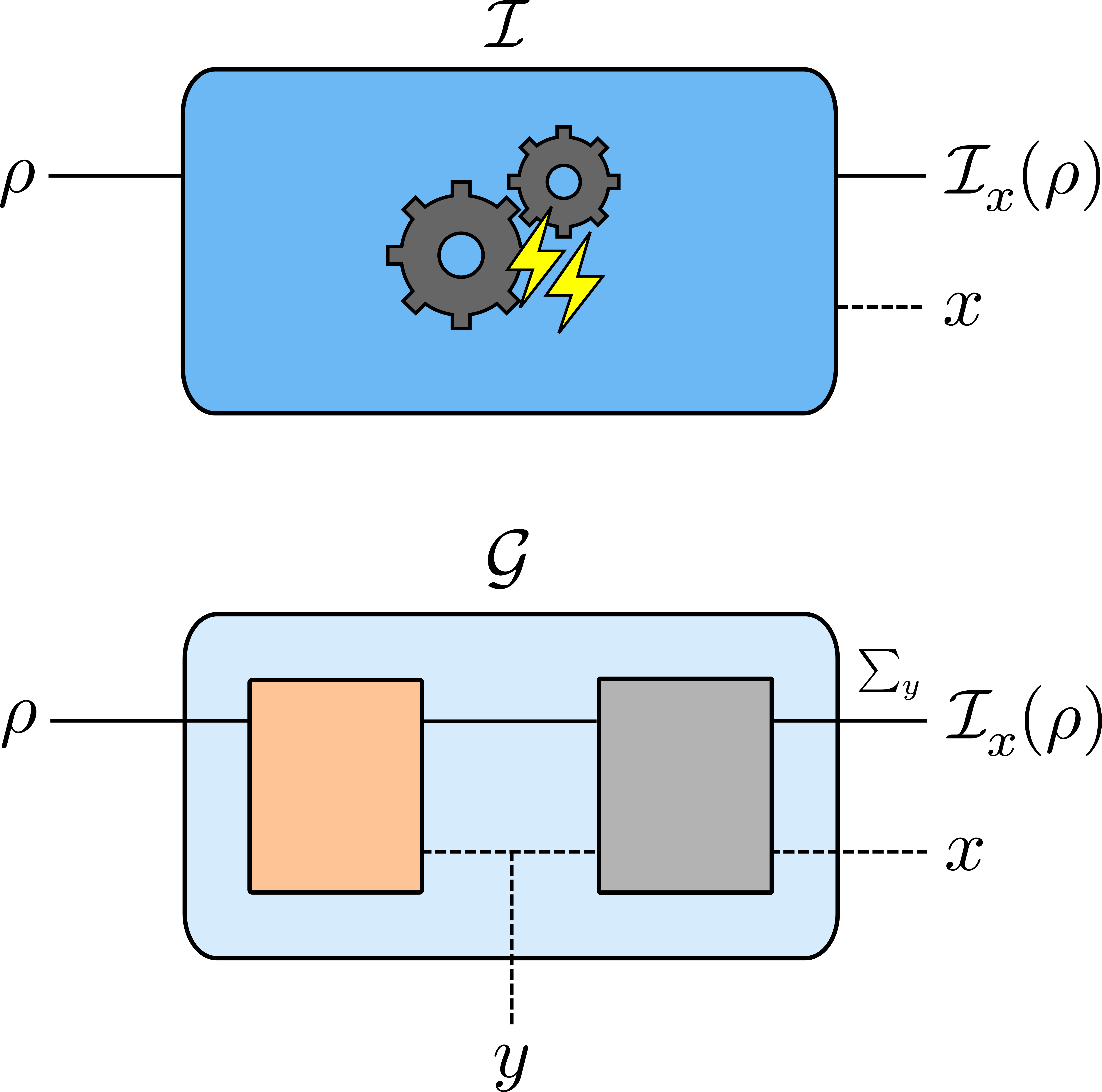}
    \caption{Observer and adversary viewpoints. In the top figure, the quantum instrument $\I$ corresponds to a black-box: the observer has full knowledge of the outputs of the measurement process for any given state, but possesses no information about the underlying mechanisms by which these transformations are performed. In the bottom figure, the full joint instrument $\G$ is shown: the adversary obtains the measurement outcome $y$ from meter $\A$. By taking the marginals over the adversary's outcome, the observer's expected outcomes are obtained.}
    \label{fig:obs_adv_1}
\end{figure}

%%%%%%%%%%%%%%%%%%%%%%%%%%%%%%%%%%%%%%%%%%%%%%%%%%%%
\section{Maximal adversarial information gain in qubit measurements}
%%%%%%%%%%%%%%%%%%%%%%%%%%%%%%%%%%%%%%%%%%%%%%%%%%%%

In this section, we get back to the introduced framework of information retrieval.
In the current setting $\mathcal{I}$ is in the possession of the observer, while $\A$ is used by the adversary. 
In order to say anything more specific about the compatibility relation, we need to specify some classes of instruments and meters.
In the following, we focus on qubit measurements.
We define a noisy unbiased dichotomic qubit meter $\A^{t, \hat{m}}$ as
\begin{align}
\A^{t, \hat{m}}(\pm) \coloneq \frac{1}{2} \left( \id \pm t \sigma_{\hat{m}} \right), 
\end{align}
where $0\leq t \leq 1$, $\hat{m}\in\real^3$ is a unit vector and $\sigma_{\hat{m}} = m_x\sigma_x + m_y \sigma_y + m_z \sigma_z$. 
This is a sharp (i.e. projective) meter if and only if $t=1$, and a trivial coin-tossing meter if and only if $t=0$.
Generally, for $t<1$, these are noisy meters in the sense of Example \ref{ex:noise}.

A class of instruments that describes different implementations of this meter is
\begin{align}\label{eq:mixed_inst}
    \I_{\pm}^{\lambda, t, \hat{m}}(\varrho) \coloneq \  \lambda \sqrt{\A^{t, \hat{m}}(\pm)} \varrho \sqrt{\A^{t, \hat{m}}(\pm)} \ + (1-\lambda) \tr{\A^{t, \hat{m}}(\pm) \varrho} \frac{\id}{2}.
\end{align}
There is now an additional noise parameter $0 \leq \lambda \leq 1$. We can justify this noise model as follows: As was seen in Example \ref{ex:noise}, in a non-adversarial scenario the noise in the meter can be understood to originate from the depolarizing channel $\mathcal{D}_t$  since $\A^{t, \hat{m}}(\pm) = \mathcal{D}_t^*\left(\A^{1, \hat{m}}(\pm)\right)$, or since $\A^{t, \hat{m}}$ is unbiased, the noise can be considered equally well to originate from the classical postprocessing $\nu^t$ as $\A^{t, \hat{m}}(\pm) = (\nu^t \circ \A^{1, \hat{m}})(\pm)$. 
Then, the noise for the full instrument can be understood to originate from the depolarizing channel and/or classical postprocessing as
\begin{align}
    \I_{\pm}^{\lambda, t, \hat{m}} = \mathcal{D}_\lambda\circ \luders{\mathcal{D}_t^*( \A^{1, \hat{m}})}_{\pm} = \mathcal{D}_\lambda\circ \luders{\nu^t \circ \A^{1, \hat{m}}}_{\pm}.
\end{align}
In the case $\lambda=1$, the instrument is the Lüders instrument of $\A^{t, \hat{m}}$.
In the other extreme case $\lambda=0$, the instrument is a measure-and-prepare instrument.
The parameters $t$ and $\lambda$ are independent, which relates to the physical fact that each meter has multiple implementations\footnote{Assuming that depolarizing channels accurately represent the underlying noise of both the meter and the instrument, the parameters $t$ and $\lambda$ could be experimentally obtained by computing the Pauli transfer matrix of each operation through methods like quantum instrument linear gate set tomography (See \cite{Rudinger2022, Hashim2025, Hashim2026})}.

As discussed in the introduction, in the two special cases the adversarial information gain is simple to solve: 
\begin{itemize}
\item When $\lambda=1$, then $\A^{s,\hat{n}}$ is compatible with $\I^{1, t, \hat{m}} = \luders{\A^{t,\hat{m}}}$ if and only if $\A^{s,\hat{n}}$ can be classically post-processed from $\A^{t,\hat{m}}$, as shown in \cite[Cor. 7]{Leppjrvi2024}. This is possible if and only if $\hat{n} = \hat{m}$ and $s \leq t$. 
\item When $\lambda=0$ we have that $\A^{s,\hat{n}}$ is compatible with the measure-and-prepare instrument $\I^{0, t, \hat{m}}$ if and only if $\A^{s,\hat{n}}$ and $\A^{t,\hat{m}}$ are compatible meters \cite[Prop. 9]{Leppjrvi2024}. This is possible if and only if \cite[Thm. 4.5]{Busch1986-rx}
\begin{equation}\label{eq:qubit_comp}
    ||s \hat{n}+ t \hat{m}||+||s \hat{n}- t \hat{m}||\leq 2.
\end{equation}
In the special case when the directions are complementary, i.e. $\hat{n} \perp \hat{m}$, this reduces to the condition $s^2+t^2\leq 1$.
\end{itemize}

In the cases $0 < \lambda < 1$, the adversarial information gain is more involved.
It drastically depends on the type of information that the adversary tries to acquire.
In the following, we characterize the maximal information gain in the cases of aligned directions (i.e. $\hat{n} = \hat{m}$) and complementary directions (i.e. $\hat{n} \perp \hat{m}$).
The following theorems characterize the compatibility of the instrument $\mathcal{I}^{\lambda, t, \hat{m}}$ with the previously described two classes of meters $\A^{s, \hat{n}}$.

\begin{theorem}\label{thm:aligned}\emph{(Aligned directions)}
An instrument $\I^{\lambda, t, \hat{m}}$ and a meter $\A^{s, \hat{m}}$ are compatible if and only if
\begin{equation}
s \leq \frac{1}{2}\left[1- \lambda + \sqrt{(1-\lambda)(1+3\lambda) + 4t^2 \lambda^2} \right].
\end{equation}
\end{theorem}

\begin{theorem}\label{thm:comp}\emph{(Complementary directions)}
An instrument $\I^{\lambda, t, \hat{m}}$ and a meter $\A^{s, \hat{n}}$, where $\hat{n} \perp \hat{m}$, are compatible if and only if
\begin{equation}
s \leq  \frac{1}{2} \left[ 1- \lambda + \sqrt{(1-\lambda)(1+3\lambda)}\right]\sqrt{1-t^2}.
\end{equation}
\end{theorem}

We postpone the proofs of these theorems to the next section. Before that, we discuss the consequences.
The regions defined in the theorems are depicted in Fig. \ref{fig:region-comparison}.

\begin{figure}[h]
    \centering
    \begin{subfigure}{0.49\linewidth}
        \centering
        \includegraphics[width=\linewidth]{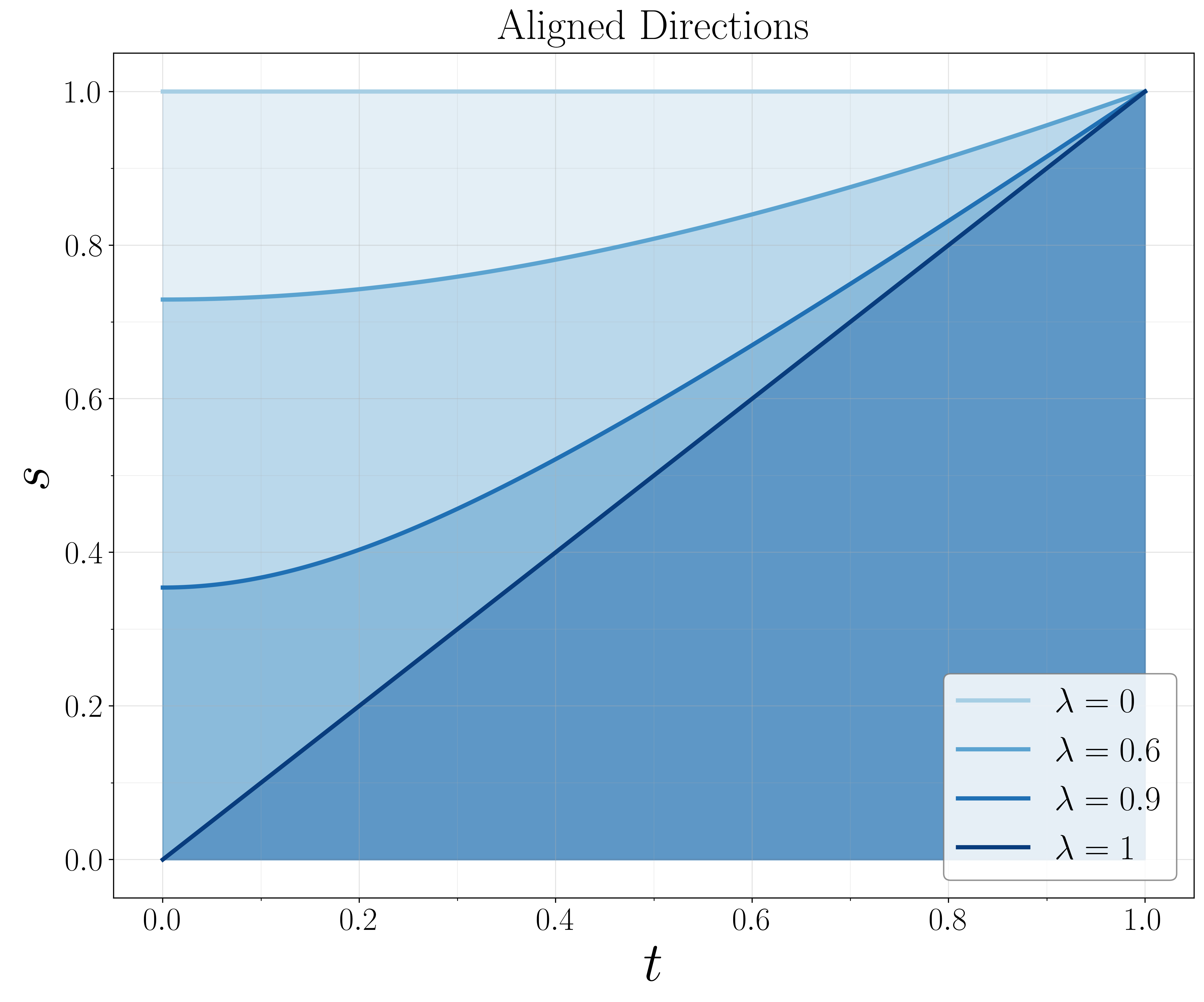}
    \end{subfigure}
    \hfill
    \begin{subfigure}{0.49\linewidth}
        \centering
        \includegraphics[width=\linewidth]{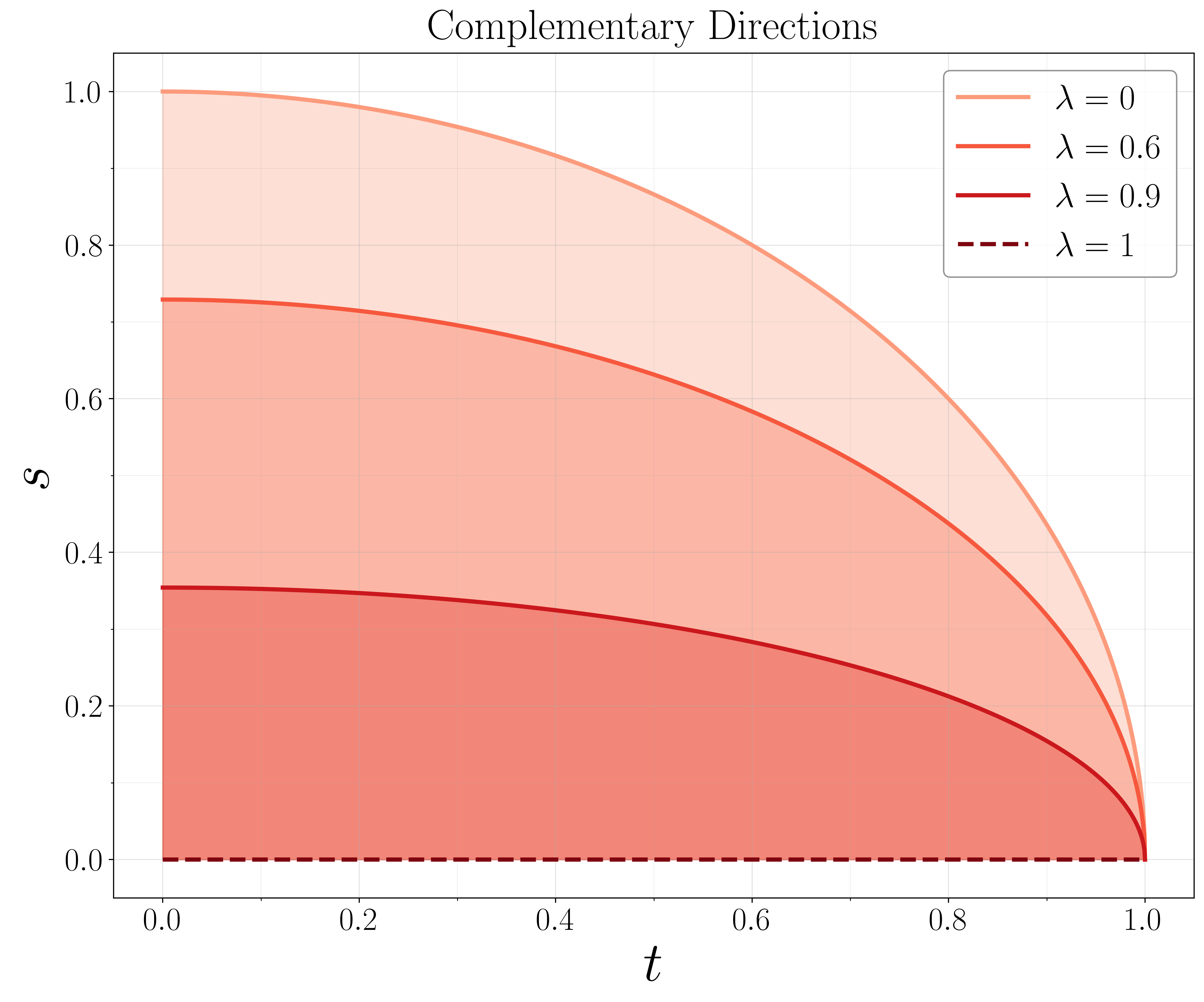}
    \end{subfigure}
    \caption{Regions of compatibility for instrument $\mathcal{I}^{\lambda,t,\hat{m}}$ and meter $\mathcal{A}^{s,\hat{n}}$ for different noise values $\lambda$, in the aligned (left) and complementary (right) directions.}
    \label{fig:region-comparison}
\end{figure}

Firstly, notice that we recover the adversarial information gain from the previously mentioned special cases: 

\begin{itemize}
    \item For $\lambda=1$, the adversary's information gain satisfies $s\leq t$ in the aligned directions case (dark blue area), while in the complementary directions the adversary cannot retrieve information (black dotted line).
    \item For $\lambda=0$, the adversary's information gain satisfies $s \leq 1$ in the aligned directions (light blue area), while in the complementary directions, it satisfies $t^2+s^2\leq1$ (light red area); both inequalities can be recovered from Eq.~\eqref{eq:qubit_comp}.
\end{itemize} 

Notice that whenever $t=0$, the direction becomes irrelevant, which can be observed as both inequalities coincide: 

\begin{equation}
    s \leq \frac{1}{2}\left[1- \lambda + \sqrt{(1-\lambda)(1+3\lambda)} \right].
\end{equation}

This inequality corresponds to the necessary and sufficient condition of compatibility between the noisy qubit meter $\A^{s,\hat{n}}$ and the depolarizing channel $\mathcal{D}_\lambda$ \cite{Heinosaari2018,Torii2025}. Indeed, in our case when $t=0$ the instrument $\I^{\lambda, 0, \hat{m}}$ consists of two quantum operations that are proportional to the depolarizing channel $\mathcal{D}_\lambda$.

\begin{corollary}\label{cor:max}
The maximum amount of information that the adversary can obtain, quantified by the sharpness $s_{\max}$ of the meter $\A^{s_{\max}, \hat{n}}$, for the aligned and complementary directions is, respectively

\begin{align}
s^{a}_{\max} &= \frac{1}{2}\left[1- \lambda + \sqrt{(1-\lambda)(1+3\lambda) + 4t^2 \lambda^2} \right],\\
s^{c}_{\max} &= \frac{1}{2} \left[ 1- \lambda + \sqrt{(1-\lambda)(1+3\lambda)}\right]\sqrt{1-t^2}.
\end{align}
\end{corollary}

The maximal information gain in the two situations of Cor. \ref{cor:max} is depicted in Fig. \ref{fig:smax-comparison}. For a fixed value of $\lambda$, as the sharpness $t$ of the observer's meter increases, we notice that the adversary's information gain can either increase (aligned case), or decrease (complementary case). On the other hand, for a fixed value of $t$, as the parameter $\lambda$ decreases, the adversary's information gain increases in both cases. The big leap from $\lambda=1$ to $\lambda=0.9$ shows that a small amount of noise already gives a significant difference in information gain.

\begin{figure}[h]
    \centering
    \begin{subfigure}{0.49\linewidth}
        \centering
        \includegraphics[width=\linewidth]{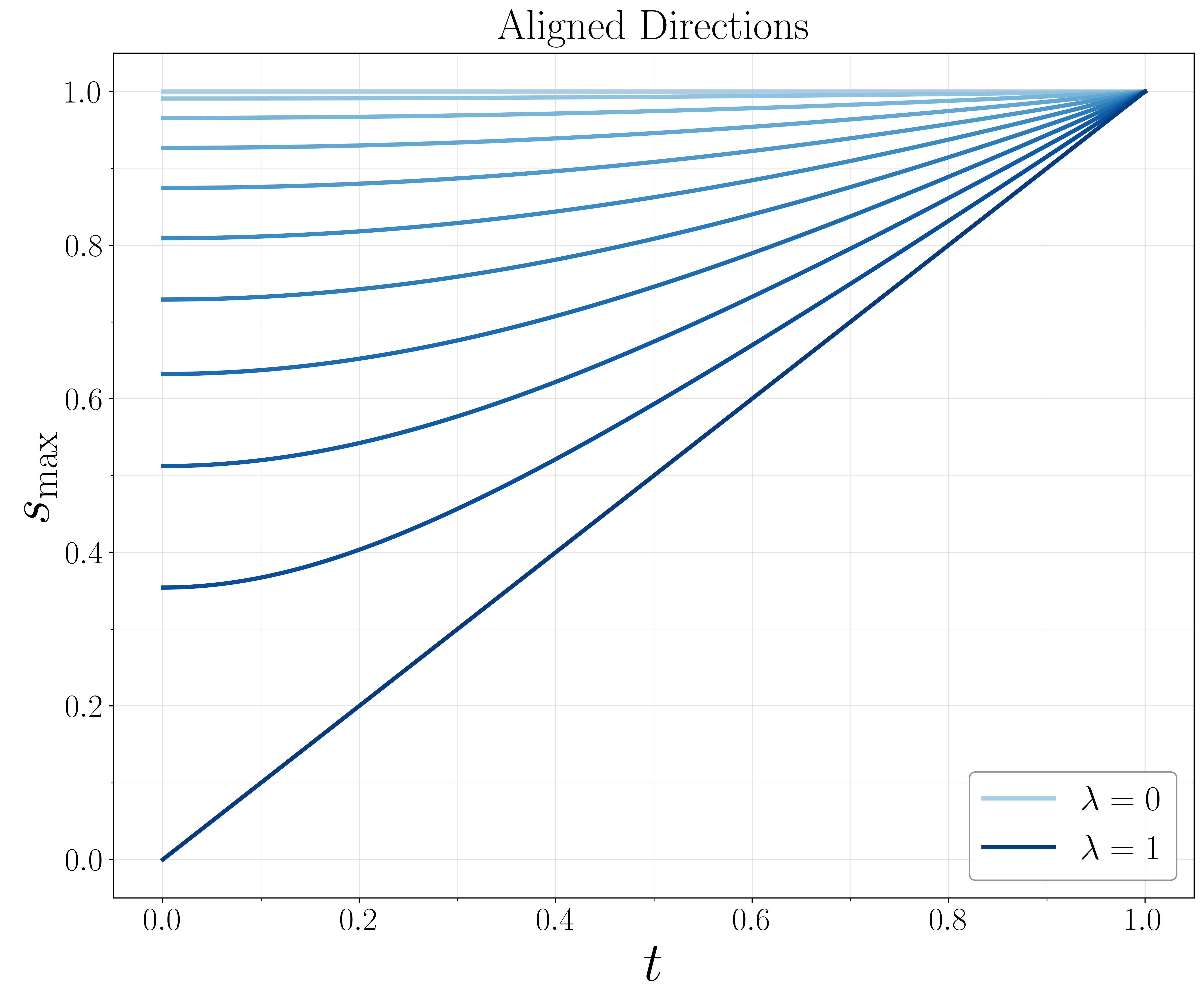}
    \end{subfigure}
    \hfill
    \begin{subfigure}{0.49\linewidth}
        \centering
        \includegraphics[width=\linewidth]{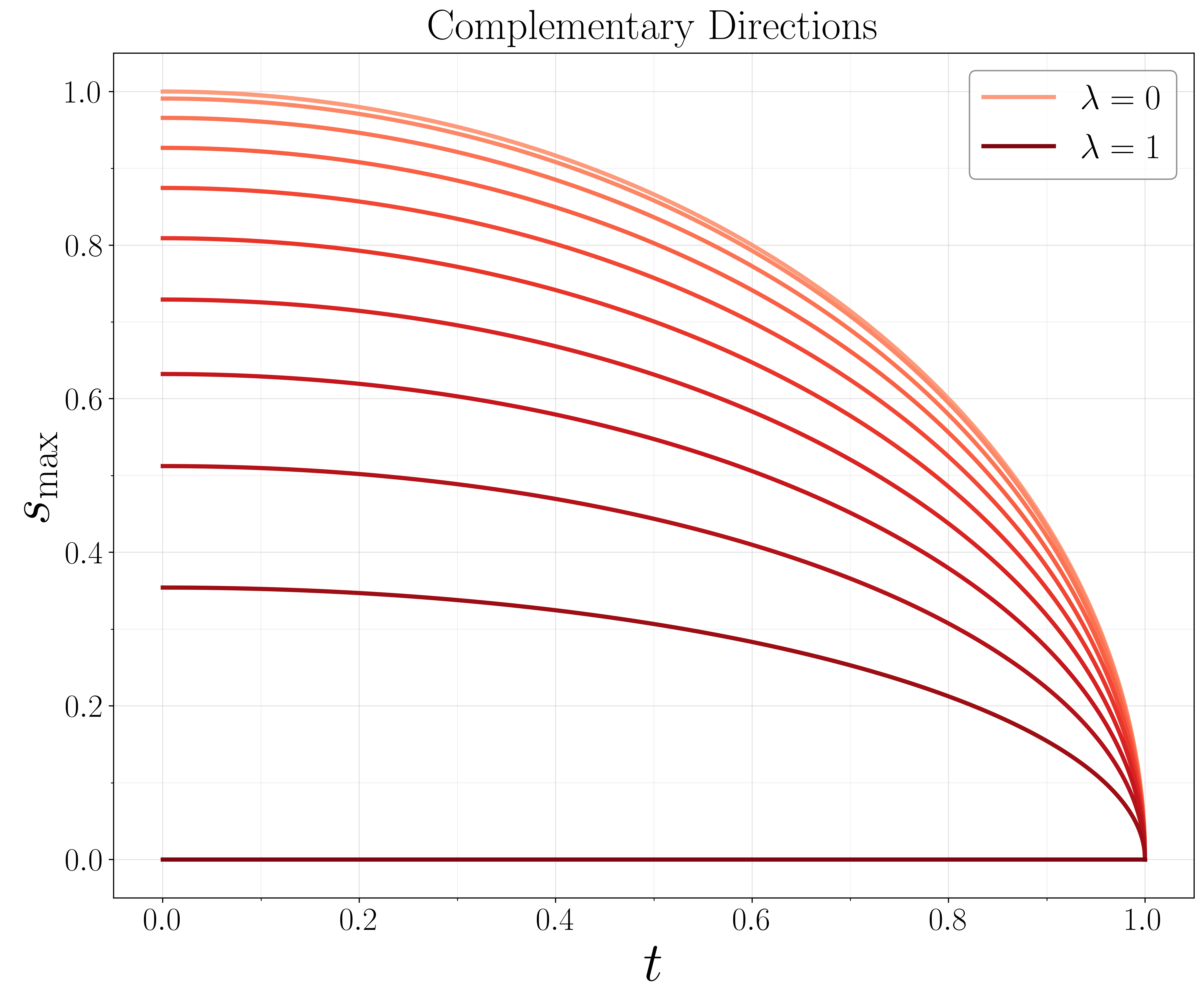}
    \end{subfigure}
    \caption{Adversary's maximal information gain $s_{\max}$ as a function of the sharpness $t$ for eleven equispaced values of $\lambda$ from 0 to 1, in the aligned (left) and complementary (bottom) directions.}
    \label{fig:smax-comparison}
\end{figure}

%%%%%%%%%%%%%%%%%%%%%%%%%%%%%%%%%%%%%%%%%%%%%%%%%%%%
\section{Compatibility of a meter and an instrument}
%%%%%%%%%%%%%%%%%%%%%%%%%%%%%%%%%%%%%%%%%%%%%%%%%%%%

As we have seen in the previous section, the crucial part in understanding the limits of adversarial information gain is to analyze the meter-instrument compatibility.
In order to give a useful characterization for meter-instrument compatibility, we need to look more carefully into how instruments can be implemented.

It can be shown \cite{Ozawa1984,Chiribella2009} that for every instrument $\I \in \ins_n(\hi)$ there exists a dilation $(\hi_A, W,\Eo)$, that is, an ancillary Hilbert space $\hi_A$, an isometry $W : \hi \rightarrow \hi_A \otimes \hi$ and a meter $\Eo \in \meters_n(\hi_A)$ such that:
\begin{align}\label{eq:dil}
    \I_x(\varrho) = \ptr{\hi_A}{W \varrho W^{*} (\Eo(x) \otimes \id_\hi)}
\end{align}
for all $\varrho \in \sh$ and $x \in [n]$. We note that the action of the isometry can also be expressed in terms of a unitary interaction between the system and some fixed state of the ancillary system. In this context (especially in the case when the meter $\Eo$ is a projective measurement), a dilation is often referred to as a measurement model \cite{Ozawa1984,Heinosaari2011} since it gives a description of how the instrument can be realized in terms of letting the system interact with some ancillary system and then making a measurement on the ancillary state.

In Eq.~\eqref{eq:dil}, the instrument is obtained by performing a demolishing measurement on the ancillary system, which includes tracing out the ancilla. If instead we trace out the output system of the original instrument, we get the conjugate instrument $\overline{\I} \in \ins_n(\hi, \hi_A)$ relative to this dilation \cite{Leppjrvi2024}. The conjugate instrument $\overline{\I}$ is then given by
\begin{align}
\overline{\I}_x(\varrho) = \ptr{\hi}{(\sqrt{\Eo(x)} \otimes \id_\hi) W \varrho W^{*} (\sqrt{\Eo(x)} \otimes \id_\hi)}
\end{align}
for all $\varrho \in \sh$ and $x \in [n]$. Notice that the conjugate instrument can be expressed as the conjugate of the induced channel $\overline{\Phi^{\I}}(\varrho) = \ptr{\hi}{W \varrho W^{*}}$ followed by the Lüders instrument of meter $\Eo$ on the ancillary system:

\begin{align}
\overline{\I}_x(\varrho) = \L_x^{\Eo} \circ \overline{\Phi^{\I}}(\varrho) = \L_x^{\Eo}(\ptr{\hi}{W \varrho W^{*}}).
\end{align}

Both the instrument dilation and the conjugate instrument are depicted in Fig.~\ref{fig:dilation}.

\begin{figure}[h]
    \centering
    \includegraphics[width=0.7\linewidth]{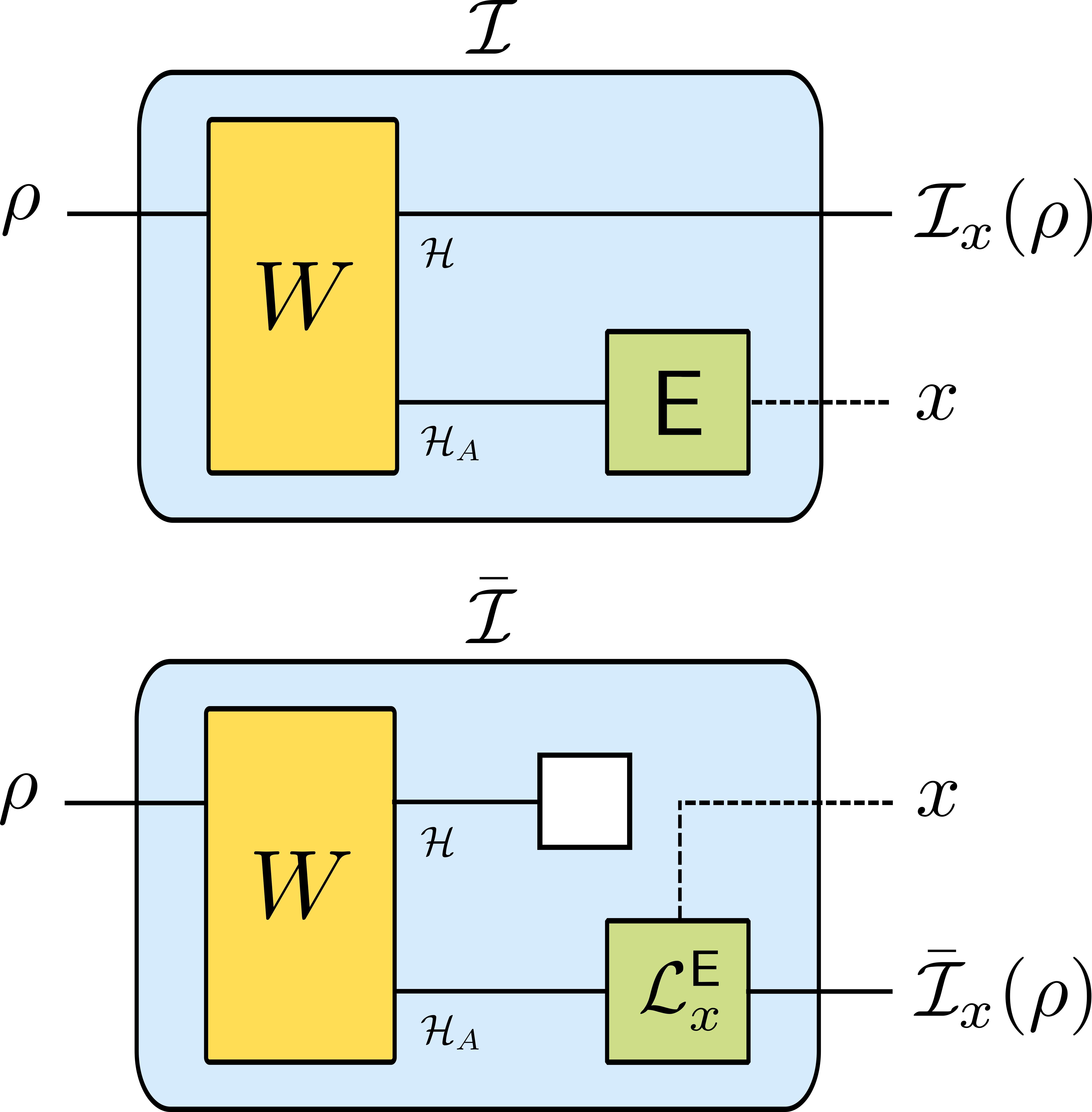}
    \caption{Top: Instrument $\I$ implemented through the dilation $(\hi_A, W,\Eo)$. The top and bottom wires correspond to the systems belonging to the output  Hilbert space $\hi$ and the ancillary Hilbert space $\hi_A$, respectively. Bottom: Conjugate instrument $\overline{\I}$ relative to the same dilation. The partial trace (white box), is taken over the output system, and the corresponding classical and quantum outcomes are obtained through the Lüders instrument $\L^{\Eo}$ acting on the ancillary system.}
    \label{fig:dilation}
\end{figure}

We are now ready to state the following general characterization of meter-instrument compatibility based on \cite{Leppjrvi2024}, which we will use later.

\begin{proposition}\label{pr:1}
A meter $\A\in \meters_m(\hi)$ is compatible with an instrument $\I \in \ins_n(\hi)$ if and only if for any conjugate instrument $\overline{\I}$ related to any dilation $(\hi_A, W,\Eo)$ of $\I$ there exists a collection of meters $\left\{\B^y \right\}_{y \in [n]} \subset \meters_n(\hi_A)$ such that
\begin{align}\label{eq:met-inst-comp}
    \A(x) = \sum_{y \in [n]} \overline{\I}^*_y(\B^y(x))
\end{align}
for all $x \in [m]$.
\end{proposition}
\begin{proof}
We know from \cite{Leppjrvi2024} that two instruments $\I \in \ins_n(\hi)$ and $\J \in \ins_m(\hi)$ are compatible if and only if for any conjugate instrument $\overline{\I}$ related to any dilation $(\hi_A, W,\Eo)$ of $\I$ there exists a collection of instruments $\left\{\R^y\right\}_{y \in [n]} \subset \ins_m(\hi_A, \hi)$ such that $\J_x = \sum_{y \in [n]} \R^y_x \circ \overline{\I}_y$ for all $x \in [m]$.

Let us consider the meter $\A$ as an instrument, that is, $\mathcal{A} \in  \ins_m(\hi, \mathbb{C})$, where
\begin{align}
    \mathcal{A} _x(\varrho) = \tr{\A(x) \varrho}
\end{align}
for all $x \in [m]$ and $\varrho \in \sh$.
Then, compatibility of $\I$ and $\mathcal{A}$ is equivalent to the existence of a set of instruments $\left\{\mathcal{B}^y\right\}_{y \in [n]} \subset \ins_m(\hi_A, \mathbb{C})$, with $\mathcal{B}^y_x =  \tr{\B^y(x) \varrho}$ and $\B^y \in \meters_m(\hi_A)$, such that
\begin{align}
    \mathcal{A}_x(\varrho) = \sum_{y \in [n]} \mathcal{B}^y_x \circ \overline{\I}_y(\varrho)
\end{align}
for all $x \in [m]$ and $\varrho \in \sh$. In the Heisenberg picture, this is equivalent to:
\begin{align}
\begin{split}
    \tr{\A(x) \varrho} &= \mathcal{A}_x(\varrho) = \sum_{y \in [n]} \mathcal{B}^y_x \circ \overline{\I}_y(\varrho) = \sum_{y \in [n]}\tr{\B^y(x)\overline{\I}_y(\varrho)}\\
    &= \sum_{y \in [n]}  \tr{\overline{\I}^*_y (\B^y(x))  \varrho} = \tr{\sum_{y \in [n]} \overline{\I}^*_y (\B^y(x))  \varrho}
\end{split}
\end{align}
for all $x \in [m]$ and $\varrho \in \sh$. As this has to be satisfied for all $\varrho \in \mathcal{S}(\hi)$, we recover Eq. \eqref{eq:met-inst-comp}.
\end{proof}

The construction of meter $\A$ as given by Proposition \ref{pr:1} is depicted in Fig.~\ref{fig:dilation_comp}.

\begin{figure}[h]
    \centering
    \includegraphics[width=0.7\linewidth]{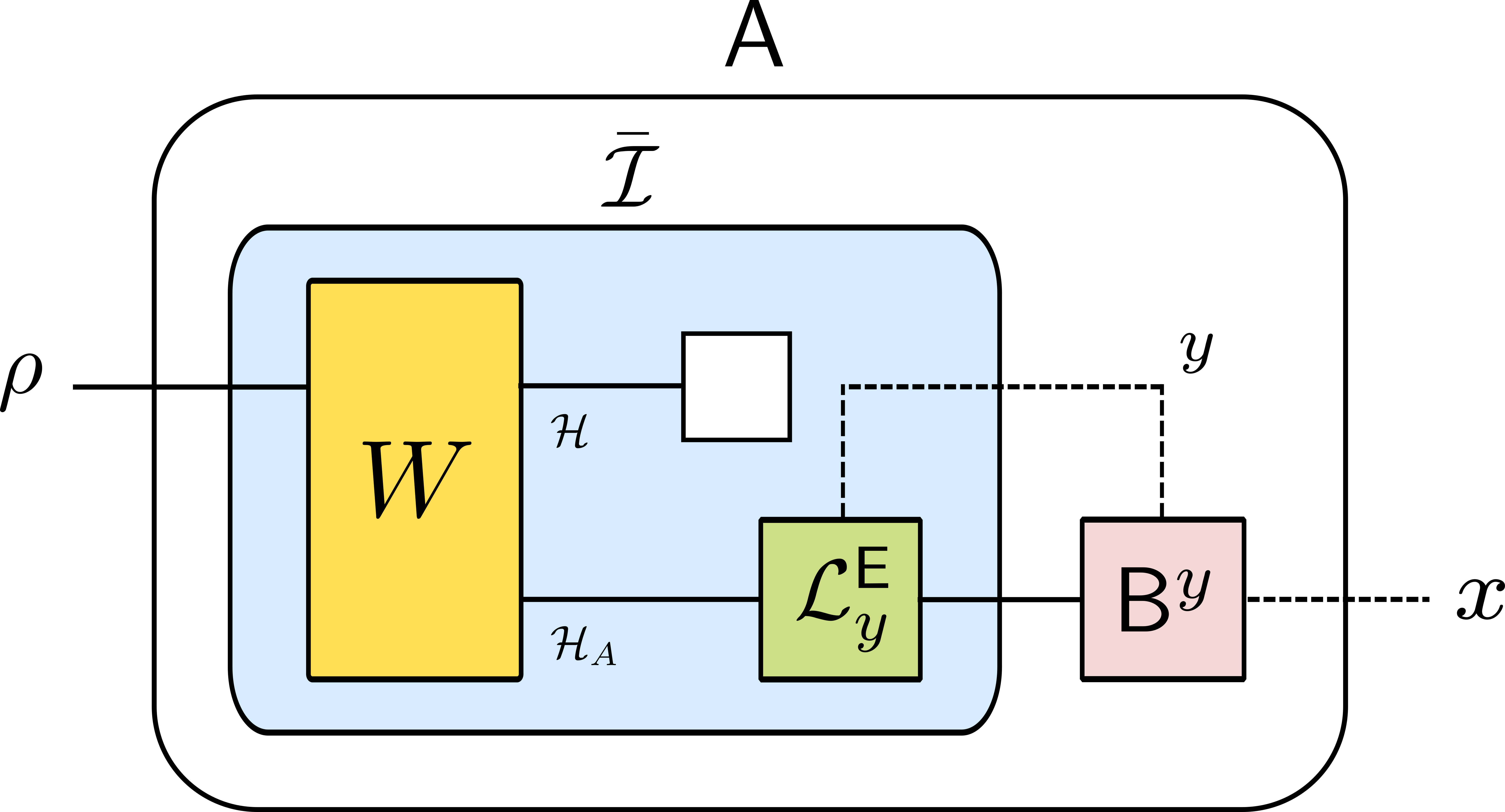}
    \caption{The compatibility of meter $\A$ and instrument $\I$ with dilation $(\hi_A, W,\Eo)$ is equivalent to the fact that $\A$ can be obtained through the conjugate instrument $\overline{\I}$ followed by meters $\B^y$ conditioned to outcome $y$.}
    \label{fig:dilation_comp}
\end{figure}

%%%%%%%%%%%%%%%%%%%%%%%%%%%%%%%%%%%%%%%%%%%%%%%%%%%%
\section{Proofs of Theorem \ref{thm:aligned} and \ref{thm:comp}}\label{sec:proofs}
%%%%%%%%%%%%%%%%%%%%%%%%%%%%%%%%%%%%%%%%%%%%%%%%%%%%

The proofs of both theorems follow similar steps to those in \cite{Heinosaari2018} in which the compatibility of an unbiased qubit meter and a Pauli channel were characterized. Let $\I^{\lambda, t, \hat{m}}$ be the instrument as defined in \eqref{eq:mixed_inst}. We fix a unit vector $\hat{n}$, and we look for the largest $s$ such that the adversary meter $\A^{s, \hat{n}}$ is compatible with $\I^{\lambda, t, \hat{m}}$. From Proposition \ref{pr:1}, we know that the instrument $\I^{\lambda, t, \hat{m}}$ and the meter $\A^{s, \hat{n}}$ are compatible if and only if there exists a collection of meters $\left\{\B^y \right\}_{y \in \left\{+,-\right\}}$ such that
\begin{align}
\A^{s, \hat{n}}(\pm) = \sum_{y \in \qty{+,-}} \qty({\overline{\I}}^{\lambda, t, \hat{m}}_y)^*(\B^y(\pm)).
\end{align}

The task mentioned previously can be stated as a semidefinite program (SDP) as follows. First, notice that
\begin{align}
\tr{\A^{s, \hat{n}}(+) \sigma_{\hat{n}}} = \tr{\frac{1}{2}(\id + s \sigma_{\hat{n}}) \sigma_{\hat{n}}} = \frac{1}{2} \tr{\sigma_{\hat{n}} + s \id} = s.
\end{align}
Then, using the previous condition and going back to the Schrödinger picture, we have
\begin{align}
\begin{split}
s &= \tr{\A^{s, \hat{n}}(+) \sigma_{\hat{n}}} = \tr{\sum_{y \in \qty{+,-}} \qty(\overline{\I}^{\lambda, t, \hat{m}}_y)^*(\B^y(+)) \sigma_{\hat{n}}}\\ 
&= \sum_{y \in \qty{+,-}} \tr{\B^y(+) \overline{\I}^{\lambda, t, \hat{m}}_y (\sigma_{\hat{n}})} = \sum_{y \in \qty{+,-}} \tr{\B^y(+) (\hat{n} \cdot \vec{\Pi}^y)},
\end{split}
\end{align}
where $\Pi^y_i \coloneq \overline{\I}^{\lambda, t, \hat{m}}_y(\sigma_i)$. Then, obtaining the largest possible value $s_{\mathrm{P}}$ such that $\A^{s, \hat{n}}$ and $\mathcal{I}^{\lambda, t, \hat{m}}$ are compatible is equivalent to maximizing the RHS of the equation over all valid effects $B^y \equiv \B^y(+)$:
\begin{align}
\begin{aligned}
s_{\mathrm{P}} \coloneq &\max_{B^+,B^-} \tr{B^+ (\hat{n} \cdot \vec{\Pi}^+) + B^- (\hat{n} \cdot \vec{\Pi}^-)},\\
& \ \text{subject to} \quad  0 \leq B^\pm \leq \id.
\end{aligned}
\end{align}

Let $\hat{n}_1,\hat{n}_2 \in \real^3$ be unit vectors such that $\hat{n}_1\perp \hat{n}_2\perp \hat{n}$. Then, we notice that
\begin{align}
\tr{\A^{s, \hat{n}}(+) \sigma_{\hat{n}_i}} = \tr{\frac{1}{2}(\id + s \sigma_{\hat{n}}) \sigma_{\hat{n}_i}} = 0, \quad \forall i \in \{1,2\}.
\end{align}
This adds two constraints to our SDP problem, mainly:
\begin{align}\label{eq:const_n}
    \tr{\A^{s, \hat{n}}(+) \sigma_{\hat{n}_i}} = \sum_{y \in \qty{+,-}} \tr{B^y (\hat{n}_i \cdot \vec{\Pi}^y)} = 0, \quad \forall i \in \{1,2\}.
\end{align}

The final primal SDP corresponds then to
\begin{align}\label{eq:primal}
\begin{split}
s_{\mathrm{P}}\coloneq\max_{B^{+},\,B^{-}} \;
    \tr{B^{+}(\hat{n}\!\cdot\!\vec{\Pi}^{+}) + B^{-}(\hat{n}\!\cdot\!\vec{\Pi}^{-})},\\
\begin{aligned}
\text{subject to} \quad & 0 \leq B^{\pm} \leq \id,\\
&\tr{B^{+}(\hat{n}_{1}\!\cdot\!\vec{\Pi}^{+})
      + B^{-}(\hat{n}_{1}\!\cdot\!\vec{\Pi}^{-})} = 0,\\
&\tr{B^{+}(\hat{n}_{2}\!\cdot\!\vec{\Pi}^{+})
   + B^{-}(\hat{n}_{2}\!\cdot\!\vec{\Pi}^{-})} = 0,\\
&\hat{n}_{1} \perp \hat{n}_{2} \perp \hat{n}.
\end{aligned}    
\end{split}
\end{align}

The Lagrangian function for the SDP is given by:
\begin{align}
\begin{split}
&\mathcal{L}(B^+, B^-, \mu_1, \mu_2, \mu_3, \mu_4, \nu_1, \nu_2)\\
& = 
  \tr{B^+ (\hat{n}\!\cdot\!\vec{\Pi}^+)} + \tr{B^- (\hat{n}\!\cdot\!\vec{\Pi}^-)} + \tr{\mu_1 \bigl(\id - B^+\bigr)} +\\
  &\quad \ \tr{\mu_2 B^+} + \tr{\mu_3 \bigl(\id - B^-\bigr)} + \tr{\mu_4 B^-} +\\
  &\quad \ \sum_{i \in \{1,2\}}
   \tr{\nu_i \left(
        -\tr{B^+ (\hat{n}_i\!\cdot\!\vec{\Pi}^+)
             + B^- (\hat{n}_i\!\cdot\!\vec{\Pi}^-)}
      \right)}                                                                    \\
&=\tr{B^+ \left[
      \left(\hat{n} - \sum_{i \in \{1,2\}} \nu_i \hat{n}_i \right)
        \!\cdot\! \vec{\Pi}^+
      - \mu_1 + \mu_2
    \right]} +\\
&\quad \  \tr{B^- \left[
      \left(\hat{n} - \sum_{i \in \{1,2\}} \nu_i \hat{n}_i \right)
        \!\cdot\! \vec{\Pi}^-
      - \mu_3 + \mu_4
    \right]} +\\
&\quad \ \tr{\mu_1 + \mu_3},
\end{split}
\end{align}
where $\mu_i$ and $\nu_i$ correspond to the dual variables of the SDP. The corresponding constraints of the dual problem become the following:
\begin{align}
\begin{split}
    \mu_i &\geq 0, \quad \forall i\in\left\{1,2,3,4\right\}\\
    \left(\hat{n}-\sum_{i \in \qty{1,2}} \nu_i \hat{n}_i \right) \cdot \vec{\Pi}^+ - \mu_1 + \mu_2 &= 0,\\
    \left(\hat{n}-\sum_{i \in \qty{1,2}} \nu_i \hat{n}_i \right) \cdot \vec{\Pi}^- - \mu_3 + \mu_4 &= 0.
\end{split}
\end{align}
Defining $\hat{l} = \hat{n}-\sum_{i \in \qty{1,2}} \nu_i \hat{n}_i$, these constraints are simplified to
\begin{align}
\begin{split}
    \mu_i &\geq 0, \quad \forall i\in\left\{1,2,3,4\right\}\\
    -\hat{l} \cdot\vec{\Pi}^+ + \mu_1 &= \mu_2,\\
    -\hat{l} \cdot\vec{\Pi}^- + \mu_3 &= \mu_4.
\end{split}
\end{align}
Then, noticing that $\hat{l} \cdot \hat{n} = 1$, rewriting $\mu_1 = \mu_+$ and $\mu_3 = \mu_-$, and the fact that $\mu_2, \mu_4$ are positive semi-definite, the final dual SDP problem can be stated as:
\begin{align}\label{eq:dual}
\begin{split}
s_{\mathrm{D}}\coloneq\min \tr{\mu_+ + \mu_-},\\
\begin{aligned}
\text{subject to} \quad &\mu_\pm \geq 0,\\
&\mu_\pm \geq \hat{l}\cdot\vec{\Pi}^\pm,\\
&\hat{l}\cdot\hat{n}=1.
\end{aligned}
\end{split}
\end{align}
We aim to find feasible points for the primal SDP and the dual SDP, and show that both yield the same value, that is, the largest possible value of $s$.

\begin{proof}[Proof Theorem \ref{thm:aligned}]

Due to the global symmetries of single qubit systems, we can choose without losing generality the observer's device to be a measurement in the $z$-direction.
This is simply a question of fixing a coordinate system.
Hence, the observer is operating with an instrument
\begin{align}\label{eq:z-inst}
\begin{split}
    \mathcal{Z}_{\pm}^{\lambda, t}(\varrho) \coloneq \I^{\lambda,t,\hat{z}}_{\pm}(\rho) = \ &\lambda \sqrt{\Z^{t}(\pm)} \varrho \sqrt{\Z^{t}(\pm)} +\\
    &(1-\lambda) \tr{\Z^{t}(\pm) \varrho} \frac{\id}{2}
\end{split}
\end{align}
with the related meter
\begin{equation}
\Z^{t}(\pm) \coloneq \A^{t,\hat{z}}(\pm)=\frac{1}{2} \left( \id \pm t \sigma_{z} \right).
\end{equation}

We start by computing the conjugate instruments related to the Pauli matrices $\Pi_i ^\pm = \overline{\mathcal{Z}}_{\pm}^{\lambda, t}(\sigma_i)$, which are given by
\begin{align}
\begin{split}
\Pi_z ^{\pm } = \frac{1}{4} \begin{pmatrix}
                                        C_\pm & 0_2\\
                                        0_2 & D_\pm
                                        \end{pmatrix}, \quad
\Pi_x ^{\pm } = \frac{\sqrt{1-\lambda}}{8} \begin{pmatrix}
                                        0_2 & E_\pm\\
                                        E_\pm^* & 0_2
                                        \end{pmatrix},\quad
\Pi_y^{\pm } = \frac{\sqrt{1-\lambda}}{8} \begin{pmatrix}
                                        0_2 & F_\pm\\
                                        F_\pm^* & 0_2
                                        \end{pmatrix},
\end{split}
\end{align}    
where
\begin{align}
\begin{split}
    C_{\pm } &= \begin{pmatrix}
        -(1 \mp t)(1-\lambda) & 0 \\
        0 & (1 \pm t)(1-\lambda) \\
    \end{pmatrix},\\
    D_{\pm } &= \begin{pmatrix}
       \pm t (1+\lambda) & -\sqrt{\Lambda+4 t^2 \lambda^2} \\
       -\sqrt{\Lambda+4 t^2 \lambda^2} & \pm t (1+\lambda) 
    \end{pmatrix},\\
    E_{\pm } &= \begin{pmatrix}
    -\sqrt{1\mp t} \ (T_{mp} \pm T_{pm}) & \sqrt{1\mp t} \ (T_{pp} \pm T_{mm}) \\
    \sqrt{1\pm t} \ (T_{mp} \mp T_{pm}) & \sqrt{1\pm t} \ (T_{pp} \mp T_{mm})
    \end{pmatrix},\\
    F_{\pm } &= \begin{pmatrix}
    -i\sqrt{1\mp t} \ (T_{mp} \pm T_{pm}) & i\sqrt{1\mp t} \ (T_{pp} \pm T_{mm}) \\
    -i\sqrt{1\pm t} \ (T_{mp} \mp T_{pm}) & -i\sqrt{1\pm t} \ (T_{pp} \mp T_{mm})
    \end{pmatrix},\\
    T_{pp} &= \sqrt{\left(1+\lambda+2 \lambda \sqrt{1-t^2}\right)\left(1+\sqrt{\frac{(1-t^2)\Lambda}{\Lambda+4 t^2 \lambda^2}}\right)},\\
    T_{pm} &= \sqrt{\left(1+\lambda+2 \lambda \sqrt{1-t^2}\right)\left(1-\sqrt{\frac{(1-t^2)\Lambda}{\Lambda+4 t^2 \lambda^2}}\right)},\\
    T_{mp} &= \sqrt{\left(1+\lambda-2 \lambda \sqrt{1-t^2}\right)\left(1+\sqrt{\frac{(1-t^2)\Lambda}{\Lambda+4 t^2 \lambda^2}}\right)},\\
    T_{mm} &= \sqrt{\left(1+\lambda-2 \lambda \sqrt{1-t^2}\right)\left(1-\sqrt{\frac{(1-t^2)\Lambda}{\Lambda+4 t^2 \lambda^2}}\right)},\\
    \Lambda &= (1-\lambda)(1+3\lambda).
\end{split}
\end{align}

Let us consider the aligned case with $\hat{n}=\hat{m}=\hat{z}$, that is, the adversary's meter corresponds to $\Z^{s}$. A feasible point for the primal problem is given by:
\begin{align}
\begin{split}
B &\coloneq B^+ = B^- = \frac{1}{2}\left[\id_4  +\frac{1}{2} \sigma_{z} \otimes \left(\sigma_{x} - \sigma_{z}\right) -\frac{1}{2} \id_2 \otimes \left(\sigma_{x} + \sigma_{z}\right)
\right].    
\end{split}
\end{align}
Indeed, the effect $B$ is a two-dimensional projection, therefore it satisfies the first two constraints of the primal SDP \eqref{eq:primal}. We obtain then a lower bound for the primal SDP, while also satisfying the last two constraints,
\begin{align}
\begin{split}
s_{\max} \coloneq \tr{B {\Pi}_z^+ + B {\Pi}_z^-}& = \frac{1}{2}\left[1- \lambda + \sqrt{\Lambda + 4t^2 \lambda^2} \right],\\
\tr{B {\Pi}_x^+ + B {\Pi}_x^-}& = 0,\\
\tr{B {\Pi}_y^+ + B {\Pi}_y^-}&= 0,
\end{split}
\end{align}
where we choose $\hat{n}_1 = \hat{x}$ and $\hat{n}_2 = \hat{y}$. We then have $s_{\max} \leq s_{\mathrm{P}}$. 

Similarly, a feasible point for the dual SDP \eqref{eq:dual} is given by 
\begin{equation}
    \mu^+ = B {\Pi}_z^+ B, \quad \mu^- = B {\Pi}_z^- B,
\end{equation}
where we choose $\hat{l} = \hat{z}$. The operators $\mu^\pm$ have two eigenvalues 0 and two eigenvalues given by 
\begin{align}
\frac{1}{4}(1\pm t)(1-\lambda) \quad  \text{and} \quad \frac{1}{4}\left[\sqrt{\Lambda+4 t^2 \lambda^2} \pm t (1+\lambda)\right],    
\end{align}
which are positive for all $\lambda, t \in [0,1]$. The operators $\mu^\pm-{\Pi}_z^\pm$ share the same eigenvalues as $\mu^\mp$, that is, with opposite signs. As they are all positive operators, all constraints of the dual problem are satisfied. The dual function then becomes

\begin{align}
\begin{split}
\tr{\mu^+ + \mu^-} &= \tr{B {\Pi}_z^+ B} + \tr{B {\Pi}_z^-B} = \tr{B {\Pi}_z^+} + \tr{B {\Pi}_z^-} = s_{\max},    
\end{split}
\end{align}
where we used the fact that $B$ is a projection. We obtain that $s_{\max} \geq s_{\mathrm{D}}$. 

As both feasible solutions lead to the same values, the choice is optimal, and we obtain that

\begin{align}
    s_{\max} = \frac{1}{2}\left[1- \lambda + \sqrt{(1-\lambda)(1+3\lambda) + 4t^2 \lambda^2} \right]
\end{align}
is both a sufficient and necessary condition. Then, for any value of $s \leq s_{\max}$, the compatibility of $\mathcal{Z}^{\lambda, t}$ and $\Z^{s}$ holds.
\end{proof}

\begin{proof}[Proof Theorem \ref{thm:comp}]

Similarly, due to the global symmetries of single qubit systems, we can consider without loss of generality the case where $\hat{n}=\hat{z}$ and $\hat{m}=\hat{x}$. Then, the observer is operating the instrument $\mathcal{Z}^{\lambda,t}$ defined in \eqref{eq:z-inst}, whereas the adversary's meter corresponds to
\begin{align}
\X^{s}(\pm) \coloneq \A^{s,\hat{x}}(\pm)=\frac{1}{2} \left( \id \pm s \sigma_{x} \right). 
\end{align}

A feasible point for the primal problem is given by:
\begin{align}
\begin{split}
B^+ &= \frac{1}{2}\left[\id_4 + c \sigma_{x} \otimes \sigma_{x} - b \sigma_{x} \otimes \sigma_{z}\right], \\
B^- &= \frac{1}{2}\left[\id_4 + b \sigma_{x} \otimes \sigma_{x} - c \sigma_{x} \otimes \sigma_{z}\right],    
\end{split}
\end{align}
where $b$ and $c$ are given by
\begin{align}
    b= \sqrt{\frac{1}{2} + \frac{t \lambda}{\sqrt{\Lambda+4 \lambda ^2 t^2}}}, \quad c=\sqrt{\frac{1}{2} - \frac{t \lambda}{\sqrt{\Lambda+4 \lambda ^2 t^2}}}.
\end{align}

The eigenvalues of $B^{\pm}$ are
\begin{align}
    \eta_1 = \frac{1}{2}(1 + \sqrt{b^2 + c^2}) = 1, \quad \eta_0 = \frac{1}{2}(1 - \sqrt{b^2 + c^2}) = 0,
\end{align}
each with multiplicity two, that is, they correspond to two-dimensional projections, satisfying the first two constraints of the primal SDP \eqref{eq:primal}. It can also be shown that the last two constraints are also satisfied,
\begin{align}
\begin{split}
s_{\max} \coloneq \tr{B^+ {\Pi}_x^+ + B^- {\Pi}_x^-}& = \frac{1}{2} \left[ 1- \lambda + \sqrt{\Lambda}\right]\sqrt{1-t^2},\\
\tr{B^+ {\Pi}_z^+ + B^- {\Pi}_z^-}& = 0,\\
\tr{B^+ {\Pi}_y^+ + B^- {\Pi}_y^-}&= 0,
\end{split}
\end{align}
where we choose $\hat{n}_1 = \hat{z}$ and $\hat{n}_2 = \hat{y}$. We then have $s_{\max} \leq s_{\mathrm{P}}$.

Similarly, a feasible point for the dual SDP \eqref{eq:dual} is given by 
\begin{align}
    \mu^+ = B^+ {\Pi}_x^+ B^+, \quad \mu^- = B^- {\Pi}_x^- B^-,
\end{align}
where we choose $\hat{l} = \hat{x}$. We need to check whether the dual constraints are satisfied, i.e. all eigenvalues of $\mu^\pm$ and $\mu^\pm-{\Pi}_x^\pm$ are positive. Upon expressing the operators ${\Pi}_x^\pm$ in the eigenbasis of the corresponding operator $B^\pm$, we notice that they have a block-diagonal form, that is:
\begin{align}
\begin{split}
{\Pi}_x^\pm = \begin{pmatrix}
M_\pm & 0 \\
0 & -M_\pm^*
\end{pmatrix}, \quad
M_\pm = m\begin{pmatrix}
g_+ & \mp\sqrt{2} t \lambda \sqrt{1+\lambda-\sqrt{\Lambda}}\\
\mp\sqrt{2} t \lambda \sqrt{1+\lambda-\sqrt{\Lambda}} & g_-
\end{pmatrix}, 
\end{split}
\end{align}
where
\begin{align}
\begin{split}
    m=\frac{1}{4}\sqrt{\frac{{(1-t^2)(1-\lambda)}}{\Lambda+4 \lambda ^2 t^2}}, \quad g_\pm = \sqrt{(1+\lambda) \left(\Lambda +2 \lambda ^2 t^2\right)+2 \sqrt{\Lambda } \lambda  \left(\lambda  t^2\pm\sqrt{\Lambda +4 \lambda ^2 t^2}\right)}.    
\end{split}
\end{align}

The dual variables can thus be rewritten as
\begin{align}
    \mu^\pm=\begin{pmatrix}
M_\pm & 0\\
0 & 0
\end{pmatrix}, \quad \mu^\pm-\Pi_x^\pm =\begin{pmatrix}
0 & 0\\
0 & M_\pm^*
\end{pmatrix}.
\end{align}
Then, all dual variables are positive if the matrices $M_\pm$ are positive semi-definite. It can be checked that $g_\pm \geq 0$ and $g_+ g_-\geq 2 t^2 \lambda^2 (1+\lambda - \sqrt{\Lambda})$ for all $t,\lambda \in [0,1]$, then by Silvester's criterion, the matrices $M_\pm$ are positive semi-definite. The dual function now corresponds to
\begin{align}
\begin{split}
\tr{\mu^+ + \mu^-} = \tr{B^+ {\Pi}_x^+ B^+} + \tr{B^- {\Pi}_x^- B^-} = \tr{B^+ {\Pi}_x^+} + \tr{B^- {\Pi}_x^-}= s_{\max},    
\end{split}
\end{align}
where we used the fact that $B^\pm$ are projections. We obtain that $s_{\max} \geq s_{\mathrm{D}}$. 

As both feasible solutions lead to the same values, the choice is optimal, and we obtain that
\begin{align}
    s_{\max} = \frac{1}{2} \left[ 1- \lambda + \sqrt{(1-\lambda)(1+3\lambda)}\right]\sqrt{1-t^2}
\end{align}
is both a sufficient and necessary condition. Clearly, for any value of $s \leq s_{\max}$, the compatibility of $\mathcal{Z}^{\lambda, t}$ and $\X^{s}$ also holds.

\end{proof}

%%%%%%%%%%%%%%%%%%%%%%%%%%%%%%%%%%%%%%%%%%%%%%%%%%%%
\section{Implementation of adversarial devices}
%%%%%%%%%%%%%%%%%%%%%%%%%%%%%%%%%%%%%%%%%%%%%%%%%%%%

Having established the maximum amount of information that the adversary can obtain, we show now how the device is implemented from the adversary's point of view. As the observer's instrument $\I \in \ins_n(\hi)$ and the adversary's meter $\A\in \meters_m(\hi)$ are compatible, it follows from \cite[Cor. 3]{Leppjrvi2024} that there exists a collection of instruments $\left\{\R^{(y)}\right\}_{y \in [m]} \subset \ins_m(\hi)$ such that

\begin{equation}\label{eq:adv_impl}
    \I_x = \sum_y \R^{(y)}_x \circ \L^{\A}_y ,
\end{equation}
for all $x \in [n]$, where $\L^{\A}\in \ins_m(\hi)$ is the Lüders instrument of the meter $\A$, as shown in Fig.~\ref{fig:joint_inst}.

\begin{figure}[h]
    \centering
    \includegraphics[width=0.75\linewidth]{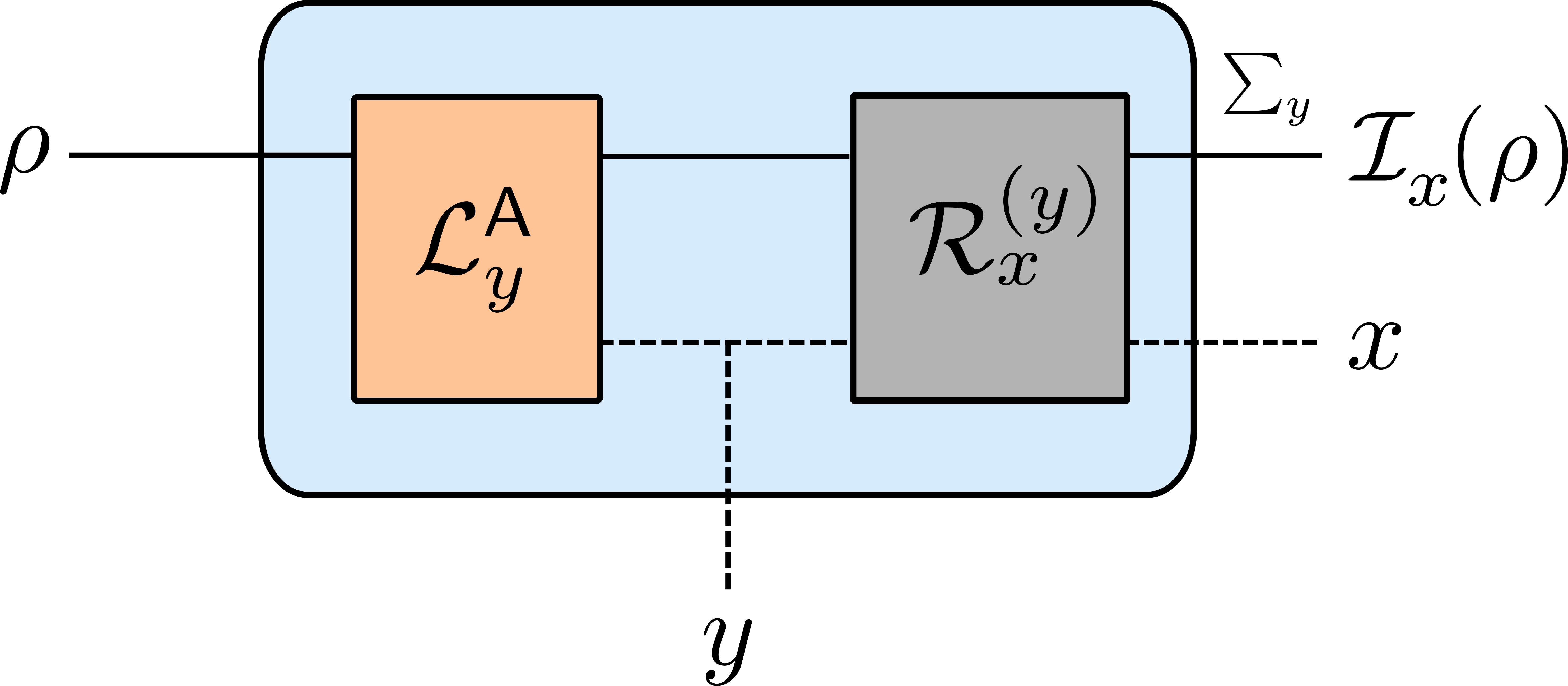}
    \caption{Adversary's implementation of the device: the adversary first retrieves information via the Lüders instrument $\L^{\A}$ of the meter $\A$, and performs a postprocessing instrument $\R^{(y)}$ according to the measurement outcome $y$. As the observer lacks knowledge of the adversary's outcome, summing over all possible outcomes $y$ returns the observer's expected classical and quantum outcomes.}
    \label{fig:joint_inst}
\end{figure}

For the aligned case, again without loss of generality, we consider
\begin{align}
\begin{split}
\mathcal{Z}_{\pm}^{\lambda, t}(\varrho) &= \lambda \sqrt{\Z^{t}(\pm)} \varrho \sqrt{\Z^{t}(\pm)} + (1-\lambda) \tr{\Z^{t}(\pm) \varrho} \frac{\id}{2}, \quad
\Z^{s_{\max}^a}(\pm) = \frac{1}{2} \left( \id \pm s_{\max}^a \sigma_z \right)    
\end{split}
\end{align}
as the observer's instrument and the adversary's meter, respectively, where $s_{\max}^a$ is given by Corollary \ref{cor:max}. From Eq. \eqref{eq:adv_impl}, we have that
\begin{equation}
    \mathcal{Z}_{x}^{\lambda, t} = \sum_y \R^{(y)}_x \circ \L^{\Z^{s_{\max}^a}}_y ,
\end{equation}
where $\L^{\A}$ is the Lüders instrument of the meter $\Z_{\pm}^{s_{\max}^a}$. The quantum operations $\R^{(y)}_x$ can be found by following the proof of \cite[Prop.7]{Leppjrvi2024}  and they are given by:
\begin{align}
\begin{split}
    \R^{(+)}_{\pm} (\varrho) &= \ptr{\hi_A}{(U_+ (\varrho \otimes \kb{0}{0}_{\hi_A}) U_+^*) (\id \otimes \tilde{\Z}^{(+)}(\pm))},\\
    \R^{(-)}_{\pm} (\varrho) &= \ptr{\hi_A}{(U_- (\varrho \otimes \kb{0}{0}_{\hi_A}) U_-^*) (\id \otimes \tilde{\Z}^{(-)}(\pm))},
\end{split}
\end{align}
with $\dim(\hi_A)=2$ and
\begin{align}
\begin{split}
    U_+ &= \mathsf{CX}_{10} (\sigma_x \otimes \id_{\hi_A}) \mathsf{CR}_Y(\theta) (\sigma_x \otimes \id_{\hi_A}),\\
    U_- &= \mathsf{CX}_{10} \mathsf{CR}_Y(\theta),\\
   \tilde{\Z}^{(+)}(\pm) &  = \frac{1}{2} \left[\left(1 \pm \tilde{\mu}\right) \id \pm \tilde{t} \sigma_z \right] = \tilde{\Z}^{(-)}(\mp).
\end{split}
\end{align}
The introduced parameters $\theta, \tilde{\mu}$ and $\tilde{t}$ are given by
\begin{align}
\begin{split}
   \theta = 2 \arcsin{\sqrt{\frac{1+\lambda}{1+s_{\max}^a}}},\quad
   \mu = t \left(\frac{s_{\max}^a + \lambda}{2s_{\max}^a + \lambda - 1}\right),\quad
   \tilde{t} = t \left(\frac{1-s_{\max}^a}{2s_{\max}^a + \lambda - 1}\right).
\end{split}
\end{align}

Here, $\mathsf{CR}_Y(\theta)$ corresponds to a controlled rotation around the $y$-axis on the Bloch sphere and $\mathsf{CX}_{10}$ is the controlled $X$ gate with the ancillary qubit as control and the input state as target:
\begin{align}
    \mathsf{CR}_Y(\theta) =\begin{pmatrix}
    1&0&0&0\\0&1&0&0\\
    0&0&\cos\frac\theta2&-\sin\frac\theta2\\
    0&0&\sin\frac\theta2&\cos\frac\theta2
    \end{pmatrix}, \ \mathsf{CX}_{10}= \begin{pmatrix}1&0&0&0\\0&0&0&1\\0&0&1&0\\0&1&0&0\end{pmatrix}.
\end{align}

On the other hand, the meters $\tilde{\Z}^{(\pm)}$ correspond to biased unsharp single qubit POVMs \cite{Heinosaari2011}, satisfying $|| \tilde{t} || \leq 1\pm \tilde{\mu} \leq 2 - || \tilde{t} ||$ and $|| \tilde{t} || \leq 1$.

The quantum circuit that implements each quantum operation $\R^{(y)}_x$ is depicted in Fig. \ref{fig:postpros1}.

\begin{figure}[h]
    \centering
    \includegraphics[width=0.75\linewidth]{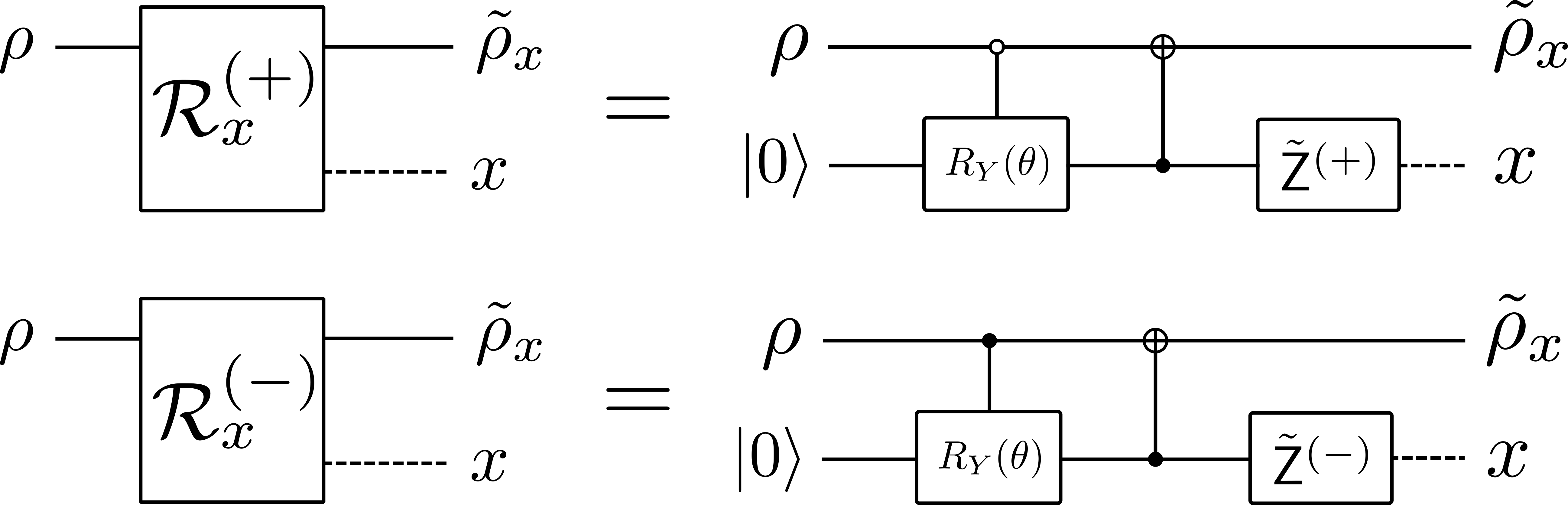}
    \caption{Quantum circuit of each post-processing quantum operation $\R^{(y)}_x$ for the aligned case.}
    \label{fig:postpros1}
\end{figure}

We notice that the induced channels of each postprocessing instrument 
\begin{equation}
    \Phi^{(\pm)} = \sum_x \mathcal{R}^{(\pm)}_x = \ptr{\hi_A}{(U_\pm (\varrho \otimes \kb{0}{0}_{\hi_A}) U_\pm^*)}
\end{equation}
have a clear operational interpretation: they correspond to amplitude damping channels \cite{Nielsen2012,Heinosaari2011}, that is, the channels $\Phi^{(\pm)}$ drive the qubit state into the corresponding ground states $P_\mp = \tfrac{1}{2} (\id \mp \sigma_z).$
%%%%%%%%%%%%%%%%%%%%%%%%%%%%%%%%%%%%%%%%%%%%%%%%%%%%
\section{Summary and conclusions}

We consider the scenario where an observer employs a device described by a non-ideal quantum instrument whose measurement outcomes and corresponding state transformations are described by unsharp meters and noisy Lüders operations, respectively. We assume that the noise on the device may arise from an adversary whose goal is to obtain non-trivial information about the input state while remaining undetectable to the observer. We aim to study how the noise of the observer's device is related to the capabilities of information gain of the adversary meter.

The aforementioned problem can be understood from the point of view of quantum compatibility: the compatibility of the adversary's meter and the observer's instrument implies the existence of a device that describes the outcomes of both parties. Then, for given noise parameters on the observer's instrument, it suffices to find conditions on the adversary's meter such that they can be implemented simultaneously. This meter represents the amount of information that he can obtain without being noticed.

For the single qubit case, we examine the scenarios where the adversary seeks to obtain information on either the same basis as the observer's measurement (aligned case) or in a mutually unbiased basis with respect to the observer's basis (complementary case). In both cases, we obtain the necessary and sufficient conditions for the compatibility of a non-ideal qubit instrument and a noisy qubit meter, which in turn allows us to calculate the maximum amount of information that the adversary can obtain in terms of the noise parameters for the observer's non-ideal instrument. As the noise in the observer's quantum operations increases, we notice that the adversary's information gain increases in both cases. As the noise in the observer's meter decreases, however, we notice that the adversary's information gain increases for the aligned case, but decreases for the complementary case.

We notice that both the meter-instrument characterization proposed in Proposition \ref{pr:1} and the subsequent SDP problem construction are equally valid for systems of dimension $d$. Nonetheless, computing the conjugate instrument becomes harder: the dimension of the ancillary Hilbert space will correspond to the Kraus rank of the induced channel, which may grow as $d$ increases. On the other hand, the SDP problem will now be constructed using $d^2-1$ generalized Gell-Mann matrices instead of Pauli matrices, which means that the number of constraints added by Eq. \eqref{eq:const_n} will scale as $d^2-2$, increasing the difficulty of finding an appropriate feasible point. A similar problem arises when considering measurements with more outcomes: increasing the number of outcomes requires us to find more suitable effects to obtain a feasible point for the SDP problem.

Finally, we provide the device implementation from the adversary's point of view. For the aligned case, the adversary needs to perform a Lüders instrument followed by a postprocessing instrument, whose induced channels correspond to amplitude damping channels. We leave the adversary's implementation for the complementary case for future work, as well as the necessary and sufficient conditions for compatibility for general directions neither aligned nor complementary.

\section*{Acknowledgments}

This work is supported by the Business Finland project BEQAH (Between Quantum Algorithms and Hardware).

%%%%%%%%%%%%%%%%%%%%%%%%%%%%%%%%%%%%%%%%%%%%%%%%%%%%

\bibliographystyle{apsrev4-2}
\bibliography{ms_References}

\end{document}